\documentclass[letterpaper, 10 pt, conference]{ieeeconf}  
\pdfoutput=1

\IEEEoverridecommandlockouts                              

\usepackage{amsthm}
\usepackage{amsmath, amsfonts, amssymb}                  
\usepackage{adjustbox}
\usepackage[font=small,belowskip=-1.5pt]{caption} 
\usepackage[colorlinks=true]{hyperref}

\usepackage{bbm}
\title{\model: Optimal and Efficient Multi-Agent Path Finding with Strategic Agents for Social Navigation}
\author {
    Rohan Chandra, Rahul Maligi,
    Arya Anantula,
    Joydeep Biswas\\{\small \texttt{\{rchandra, maligir,
    anantula.arya, joydeepb\}@cs.utexas.edu}}
    \thanks{All authors are with the Department
of Computer Science, University of Texas, Austin. Corresponding e-mail: rchandra@utexas.edu.}}

\newcommand{\model}{\textsc{SocialMapf}\xspace}

\usepackage{color}
\usepackage{xspace}
\usepackage{multicol,multirow}
\usepackage{booktabs}
\usepackage{float}
\usepackage{bm}
\usepackage{subcaption}
\usepackage[font=small]{caption}
\newtheorem{theorem}{Theorem}[section]

\newtheorem{definition}{Definition}[section]

\begin{document}

\maketitle
\thispagestyle{empty}
\pagestyle{empty}

\begin{abstract}
We propose an extension to the MAPF formulation, called \model, to account for private incentives of agents in constrained environments such as doorways, narrow hallways, and corridor intersections. \model is able to, for instance, accurately reason about the urgent incentive of an agent rushing to the hospital over another agent's less urgent incentive of going to a grocery store; MAPF ignores such agent-specific incentives. Our proposed  formulation addresses the open problem of optimal and efficient path planning for agents with private incentives. To solve \model, we propose a new class of algorithms that use mechanism design during conflict resolution to simultaneously optimize agents' private local utilities and the global system objective. We perform an extensive array of experiments that show that optimal search-based MAPF techniques lead to collisions and increased time-to-goal in \model compared to our proposed method using mechanism design. Furthermore, we empirically demonstrate that mechanism design results in models that maximizes agent utility and minimizes the overall time-to-goal of the entire system. We further showcase the capabilities of mechanism design-based planning by successfully deploying it in environments with static obstacles. To conclude, we briefly list several research directions using the \model formulation, such as exploring motion planning in the continuous domain for agents with private incentives.
\end{abstract}
\section{Introduction}
\label{sec: introduction}
The multi-agent path finding (MAPF) problem corresponds to efficiently finding optimal and collision-free paths (defined as a set of waypoints or coordinates) for $k>1$ agents in discrete $2-$D space-time. Given a bi-connected graph $\mathcal{G}$, an initial configuration $S^0$, containing the starting positions of all agents and a final configuration $G$, representing the goals of all the agents, the MAPF objective is to find the sequence of configurations, $\Gamma = \{S^0, S^1, \ldots, S^{t-1}, G\}$, where $\Gamma$ optimizes a global objective function such as the sum-of-costs or makespan.

MAPF is a fundamental problem in robotics and control and appears frequently in many real world applications such as vehicle routing~\cite{veloso2015cobots,mavrogiannis2022b,chandra2022game}, warehouse robotics~\cite{warehouse}, airport towing~\cite{airport}, and so on. Silver~\cite{silver2005cooperative} gave the first formal algorithms to solve the original MAPF problem; since then, MAPF research has included improving efficiency within bounded sub-optimality conditions, developing anytime formulations~\cite{vedder2021x}, and proposing various heuristics along these lines to improve performance. As a result of these efforts, current state-of-the-art methods~\cite{CBS, ICTS, boyarski2015icbs} and their variants can solve a wide range of simulated MAPF environments with static and dynamic obstacles. 
Since its introduction in~\cite{silver2005cooperative}, however, MAPF has always modeled cooperative agents that strive to optimize the global objective function. Consequently, planning under the classical MAPF formulation is typically centralized where agents simply execute the solution to the common global objective function.
\begin{figure}[t]
    \centering
    \includegraphics[width=\columnwidth]{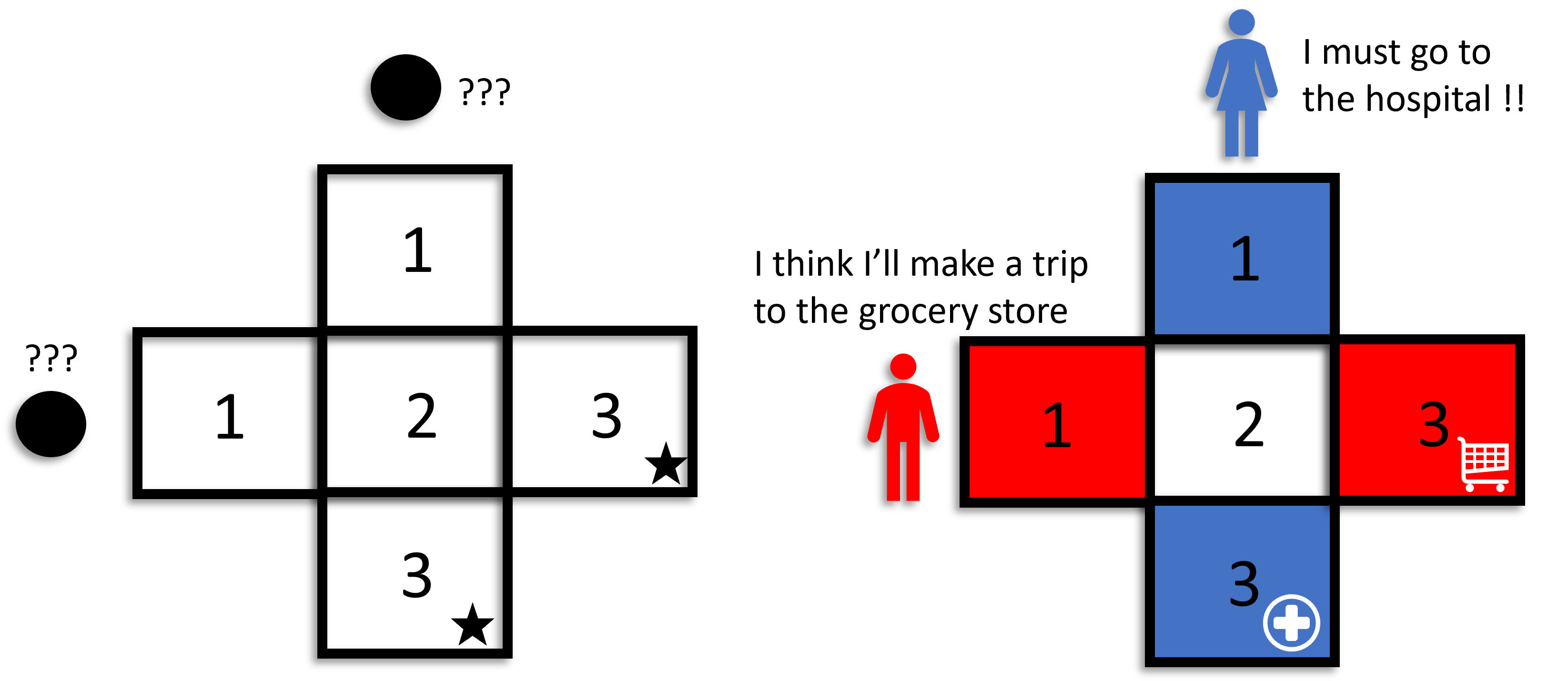}
    \caption{\textit{(left)} Classical MAPF with cooperative agents and \textit{(right)} MAPF in social navigation scenarios with strategic agents. Strategic agents have private incentives and strive to optimize individual utility functions, rather than a common global objective function.}
    \label{fig: alice_bob}
    \vspace{-5pt}
\end{figure}
\begin{figure*}[t]
\centering
   \begin{subfigure}[h]{0.19\textwidth}
    \includegraphics[width=\textwidth]{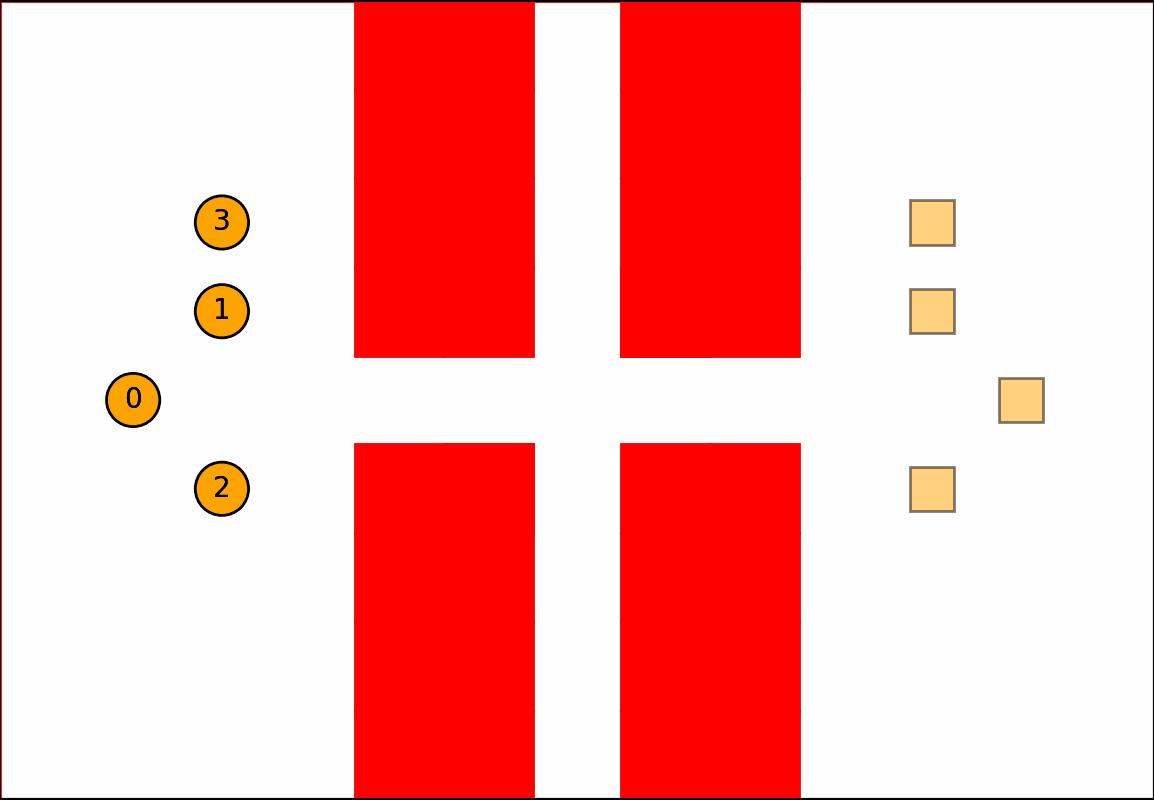}
    \caption{}
    \label{fig: doorway_agents}
  \end{subfigure}
 \begin{subfigure}[h]{0.19\textwidth}
    \includegraphics[width=\textwidth]{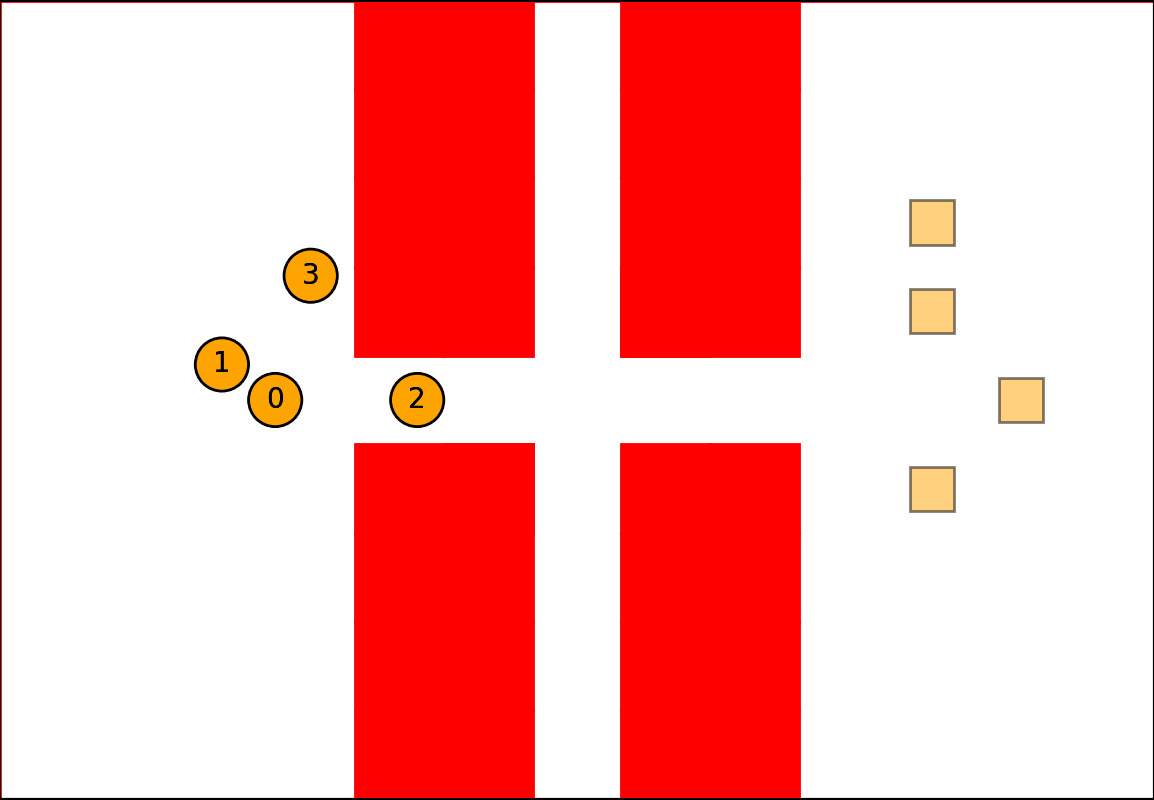}
    \caption{}
    \label{fig: hallway_agents}
  \end{subfigure}
\begin{subfigure}[h]{0.19\textwidth}
    \includegraphics[width=\textwidth]{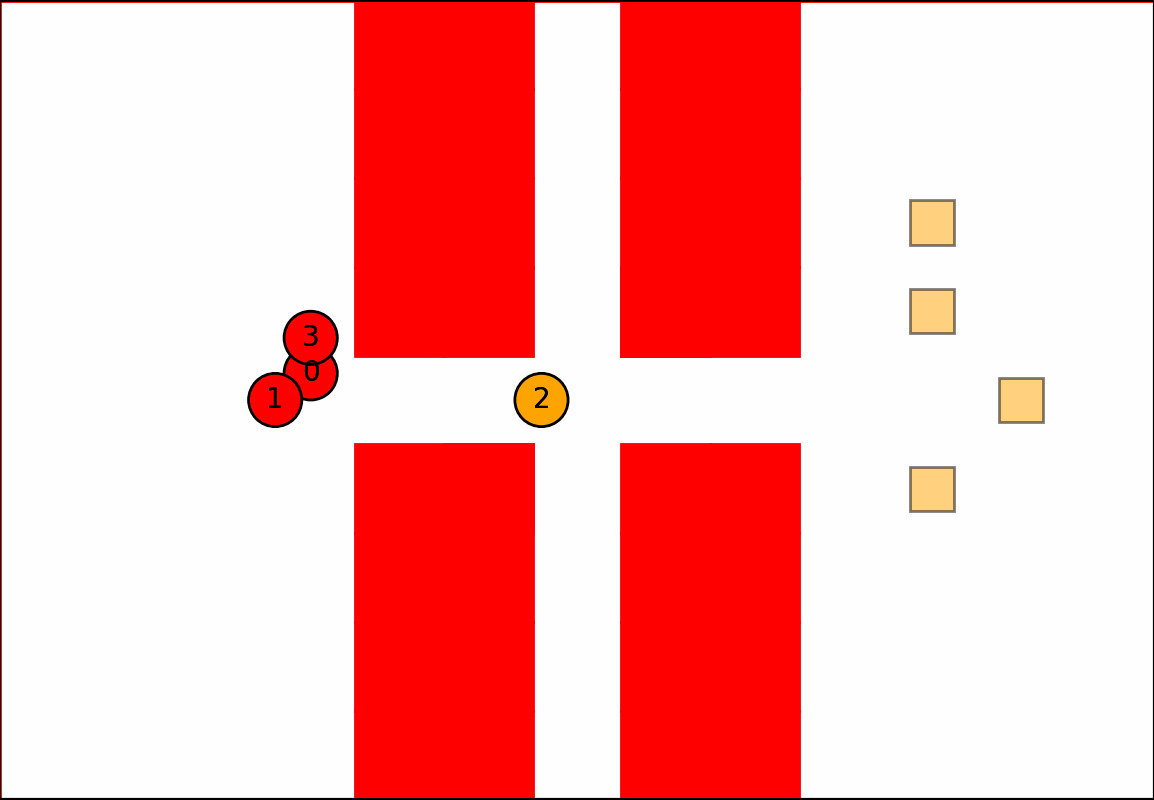}
    \caption{}
    \label{fig: intersection_agents}
  \end{subfigure}
\begin{subfigure}[h]{0.19\textwidth}
    \includegraphics[width=\textwidth]{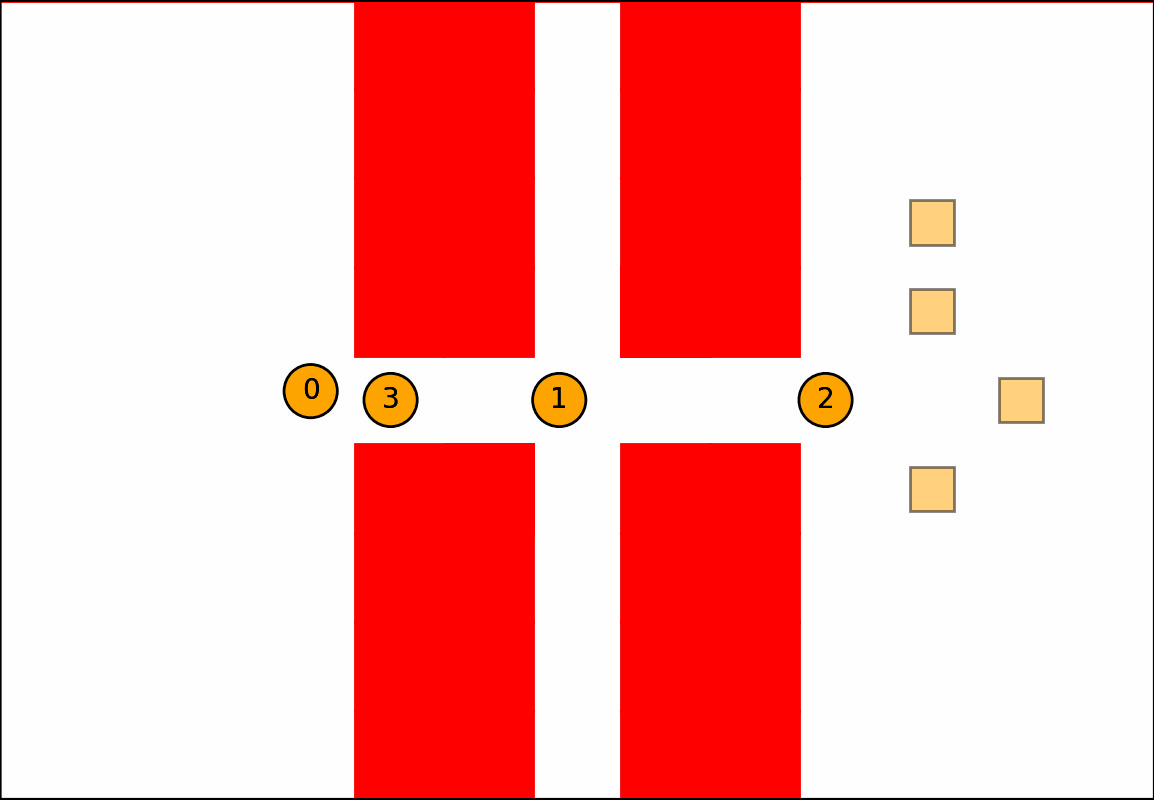}
    \caption{}
    \label{fig: doorway_gapsize}
  \end{subfigure}
\begin{subfigure}[h]{0.19\textwidth}
    \includegraphics[width=\textwidth]{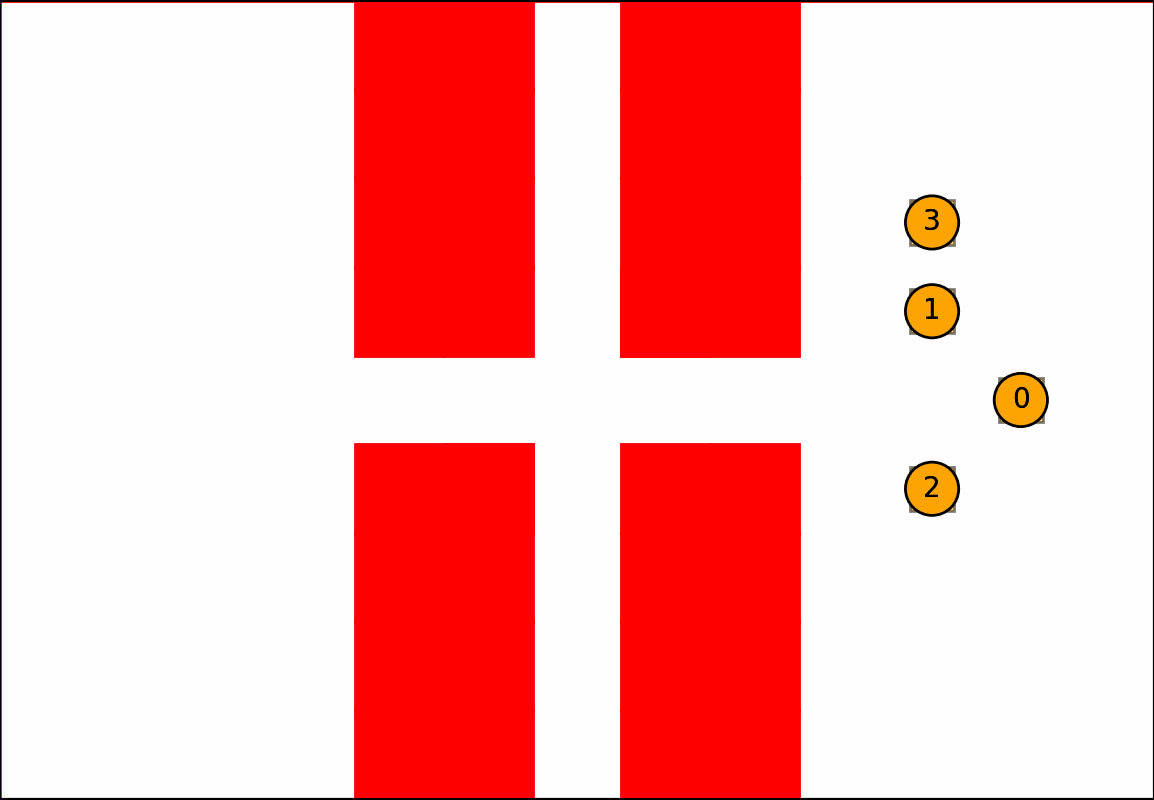}
    \caption{}
    \label{fig: hallway_gapsize}
  \end{subfigure}
\caption{Multi-agent path finding with strategic agents in a narrow hallway scenario. The yellow circles and squares denote the agents and their respective goal locations. The red regions are obstacles.}
  \label{fig: demo_figure}
  \vspace{-5pt}
\end{figure*}  

In real world social navigation scenarios, however, agents are not cooperative; they are strategic, or self-interested, possessing privately-held incentives that dictate their actions~\cite{marl, chandra2022game, forecasting}. These agents no longer share a common global objective function and instead strive to optimize utility functions unique to each agent.
Consider Figure~\ref{fig: alice_bob} depicting two different MAPF setups at a particular time $t$. On the left, we show the classical MAPF setting with two cooperative agents represented as black discs with their respective goal locations marked by $\bigstar$'s. On the right, two strategic agents, Bob (red) and Alice (blue), have distinct objectives. Alice has a higher priority of going to the hospital while Bob is making a routine trip to the grocery store. In both cases, agents can fully observe the state space and step up to $3$ tiles. In the MAPF setting, the cooperative agents can additionally choose to reveal their step size, or \textit{incentive}. On the other hand, Alice's step value is not known to Bob and vice-versa, on account of them being strategic and self-interested. The scenario on the right is what we will call a \model scenario (defined formally in Section~\ref{sec: approach}).

Centralized MAPF solvers will detect a collision at $(2,2)$ and select actions for each agent that prevent the collision. But centralized solvers cannot solve \model since strategic agents will typically not reveal their incentives. Therefore, in the \model setting in the example above, a centralized planner will yield $9$ possible pairs of actions, $(1,1), (1,2), \ldots, (3,3)$. Out of these choices, the following four--$(2,2), (2,3), (3,2), (3,3)$--lead to a collision while the remaining five are collision-free yielding a collision-free interaction with probability $\frac{5}{9}$. Sensing that the probability of \textit{not} colliding is higher than that of colliding ($>\frac{4}{9}$), centralized MAPF solvers will choose a step value of $3$ for both Alice and Bob (since that reduces the global sum-of-costs) causing them to collide. Figure~\ref{fig: demo_figure} shows an empirical simulation of a centralized MAPF solver for strategic agents.

Solving a \model problem implies finding a sequence $\Gamma_\textrm{\textsc{social}} = \{S^0, S^1, \ldots, S^{t-1}, G\}$ that optimizes both the global objective function as well as the utility functions corresponding to each strategic agent. In this work, we cast this problem through the lens of mechanism design where the designer of the mechanism proposes a set of rules that simultaneously serves to optimize both sets of objectives.





\subsection*{Main Contributions}


\begin{enumerate}
    \item Our main contribution is a new formulation called \model that extends MAPF to include strategic agents that strive to optimize local utilities as opposed to a global objective function. 
    
    \item We also introduce a class of algorithms to solve \model using mechanism design. By reducing collision resolution to a strategy-proof auction, we produce an optimal solution to \model, where optimality is defined in Definition~\ref{def: optimal_gamma}. 
    
\end{enumerate}

\noindent \model addresses several longstanding issues in MAPF such as:

\begin{itemize}
    \item \textbf{Decentralized planning among strategic agents:} In the most general sense, agents' incentives, which determine the next state or action, are hidden from other agents. Decentralized planning is often a desirable class of solutions due to their low computational requirements but require access to the agent incentives. Auction theory provides a solution for decentralized planning among agents with private incentives.
    
    
    \item \textbf{Optimality and efficiency guarantees:} 
    \model approaches are either optimal or efficient, but not both. One of the primary research areas in \model is to develop efficient solvers within bounded sub-optimality. We propose the first optimal \textit{and} efficient \model solver for real world environments with strategic agents.
\end{itemize}

In the remainder of this paper, we present the \model framework in Section~\ref{sec: framework} and propose our algorithm for solving \model in Section~\ref{sec: approach}. We present our experiments in and analyze preliminary results in Section~\ref{sec: experiments}. Finally, we conclude the paper and discuss limitations and future directions of research in Section~\ref{sec: conclusion}.
\section{Related Work}
\label{sec: related}

There is a rich history behind MAPF. Since it is impractical to discuss it all here, we refer the interested reader to a recent survey~\cite{mapf-overview}. Below, we highlight the three broad categories of solutions to discrete MAPF.

The first class of discrete MAPF algorithms consist of fast (worst-case time complexity is polynomial in the size of the graph) and complete solvers. A general template of this category of methods can be represented by the Kornhauser’s algorithm~\cite{kornhauser1984coordinating}, with a $\mathcal{O}(\lvert V \rvert^3)$ time complexity. The Push-and-Swap algorithm~\cite{pas} and its variants Parallel Push-and-Swap~\cite{ppas} and Push-and-Rotate~\cite{pnr}, improve upon Kornhauser’s algorithm in instances when there are at least two unoccupied vertices in the graph. The BIBOX algorithm is also fast and complete under these
conditions~\cite{bibox}.

The second category of methods trade efficiency for optimality guarantees. While the algorithms belonging to the previous category are fast and, under certain conditions, complete, they do not provide any guarantee regarding the quality of the solution they return. In particular, they do not guarantee that the resulting solution is optimal, either w.r.t. sum-of-costs or makespan. Extensions of A$^*$
include algorithms that search the search space using a variant of the A$^*$ algorithm. The Increasing Cost Tree Search~\cite{ICTS} splits the MAPF problem into two problems: finding the cost added by each agent, and finding
a valid solution with these costs. Finally, the Conflict-Based Search (CBS)~\cite{CBS} family solves MAPF by
solving multiple single-agent pathfinding problems. To achieve coordination, specific constraints are added incrementally to the single-agent pathfinding problems, in a way that verifies soundness, completeness, and optimality

A third category of methods consist of bounded sub-optimal
algorithms that lie in the range between the first and the second categories. An bounded sub-optimal algorithm is an algorithm that returns a solution whose cost is at most $1 + \epsilon$  times the cost of an optimal solution. Enhanced CBS (ECBS)~\cite{ecbs} and EECBS~\cite{li2021eecbs} are approximately optimal MAPF algorithms that are based on CBS.

\section{\model: Framework}
\label{sec: framework}

The two characteristics that underlie social navigation, missing from classical MAPF formulation are that $(i)$ agents have different incentives and $(ii)$ these incentives are private. Since the notion of incentives is typically ill-defined, it is best to illustrate it with examples.

\noindent \textit{Example $1$:} consider simultaneously an ambulance carrying a critical patient to the hospital and a family making a regular trip to the nearby grocery store. The ambulance is said to possess a higher ``incentive'' or priority (dictated by the urgency of the situation) to reach their goals than the family en route to the grocery store.

\noindent \textit{Example $2$:} next, consider two pedestrians arriving almost simultaneously at an indoor corridor intersection scenario. Here, unlike the previous case, the incentives of the pedestrians are hard to ascertain, but may be inferred from data using indicators like velocity, acceleration, etc.

In order to extend MAPF to social navigation, we first need to model these incentives and then figure out how to perform MAPF when the incentives of the other agents are hidden.  Table~\ref{tab:my_label} compares the MAPF and \model formulations.

\subsection{Problem formulation}

Given an arbitrary graph\footnote{In this work, we assume a $4$-connected $M\times N$ bi-connected graph.}, $\mathcal{G}$, which is discrete in space and time, at any particular time $t$, denote by $S_t$ the current configuration of $k$ \textit{strategic} agents. A strategic agent is defined as follows,

\begin{definition}
\textbf{Strategic Agent:} A strategic agent $a_i$ is a MAPF agent with a private risk function $\varphi_i:\Gamma \longrightarrow \mathbb{R}$, where $v_i = \varphi_i(\Gamma)$ encodes the $i^\textrm{th}$ agent's risk tolerance towards $\Gamma$.
\label{def: strategic_agent}
\end{definition}

\noindent We refer to $v_i$ as the private incentive of $a_i$. The \model goal is to output an optimal sequence of configurations, $\Gamma_\textsc{opt} = \{S^0, S^1, \ldots, S^{t-1}, G\}$, where $S^0$ and $G$ are the initial and final configurations, such that the $\Gamma$ satisfies a precise set of conditions that are explained later. A configuration $S_t$ at any time $t$ is specified by a list of discrete state space vectors, $s^t_i = [x^t_i, y^t_i, v_i]^\top \in \mathbb{Z}_{+}\times\mathbb{Z}_{+}\times\mathbb{Z}_{+}$ for $1\leq i \leq k$ where $x_i,y_i$ corresponds to a vertex occupied by $a_i$ in $\mathcal{G}$ and $v_i$ is the private incentive of $a_i$ encoding the priority toward its goal $g_i$. 

The discrete action space, identical for each agent, is given by $\{ \texttt{up, down, left, right, wait}\}$. At a time-step $t$ and from current vertex $[x_i, y_i]^\top$, an agent $a_i$ can select from one of these actions which we denote by $u^t_i$; for each of the first $4$ actions, each agent in \model can additionally choose to step up to $v_i$ vertices at a time. We denote the step value by $0 \leq \tau \leq v_i$. The state transition function for an agent can be specified by $s^{t+1}_i = [x^t_i \pm \tau, y^t_i, v_i]^\top$ if $u^t_i \in \{\texttt{right, left}\}$ or $s^{t+1}_i = [x^t_i, y^t_i \pm \tau, v_i]^\top$ if $u^t_i \in \{\texttt{down, up}\}$ or $s^{t+1}_i = [x^t_i, y^t_i, v_i]^\top$ if $u^t_i \in \{\texttt{wait}\}$.

\begin{table}[t]
    \centering
    \resizebox{\columnwidth}{!}{
    \begin{tabular}{lcc}
    \toprule
    Parameter & MAPF     &  \model\\
    \midrule
    State space    & $[x^t_i, y^t_i]^\top \in \mathbb{Z}_{+}^2$ & $[x^t_i, y^t_i, v_i]^\top \in \mathbb{Z}_{+}^3$\\    Action space     & \{up, down, left, right, wait\} & \{up, down, left, right, wait\}\\
    Reward     & $-$ & $\alpha_i(u^t_i) = (t_{g,i})^{-1}$  \\
    Transition function     & $\lvert s\prime - s\rvert = 1$ & $\lvert s\prime - s\rvert = \tau$ \\
    Grid specifications     & $M\times N$ & $M\times N$\\
    Global Optimality & SoC, Makespan & Social Welfare\\
    Local Optimality & $-$ & $\Phi_i(u_i, u_{-i})$\\
    \bottomrule
    \end{tabular}
    }
    \caption{Comparing the problem descriptions of MAPF and \model. $t$ represents the current time-step and $v_i$ denotes the private incentive of $a_i$.}
    \label{tab:my_label}
    \vspace{-5pt}
\end{table}


Next, we derive the global objective and agent-specific utility functions in \model. A standard global objective function for the classical MAPF problem is the sum-of-costs (SoC) given by, 

\begin{equation}
    \mathcal{F}(\Gamma) = \sum_{i=1}^{k} t_{g,i}
    \label{eq: SoC}
\end{equation}

 \noindent where $ t_{g,i}$ is the time taken by $a_i$ to reach its goal $g_i$. The classical MAPF goal is to find the optimal $\Gamma$ that minimizes $\mathcal{F}$. In our \model setting, minimizing Equation~\ref{eq: SoC} \textit{alone} does not yield an optimal solution as we need to simultaneously maximize the agents' utility. Each agent $a_i$ in \model receives a reward, $\alpha_i(u_i) = (t_{g,i})^{-1}$, on reaching the goal in time $t_{g,i}$ by executing action $u_i$. Furthermore, since agents are strategic and cannot observe the incentives of other agents, we introduce the notion of a penalty function, $p_i(\alpha_{-i}) \in \mathbb{R}$, incurred by each agent by executing $u_{-i}$. The utility function of $a_i$ can be written as,

\begin{equation}
     \Phi_i (u_i, u_{-i}) =  \alpha_i(u_i)  - p_i(\alpha_{-i}))
 \label{eq: agent_utility}
\end{equation}

\noindent At this point, we can state the definition for $\Gamma_\textsc{opt}$:

\begin{definition}
\textbf{$\Gamma_\textsc{opt}$}: A sequence of transitions, $\Gamma = \{S^0, S^1, \ldots, S^{t-1}, G\}$ is optimal when the global sum-of-cost (Equation~\ref{eq: SoC}) is minimal and local agent utilities (Equation~\ref{eq: agent_utility}) are maximal.
\label{def: optimal_gamma}
\end{definition}

\noindent Finally, we define the model for conflict resolution in \model. A conflict is defined by the tuple $\langle \mathcal{C}^t_{\mathcal{O}}, \mathcal{O}, t \rangle$, which denotes a conflict between agents belonging to the set $\mathcal{C}^t_{\mathcal{O}}$ at time $t$ over the vertex $\mathcal{O}$ in $\mathcal{G}$. Naturally, agents must either find an alternate non-conflicting path or must move through $\mathcal{O}$ in a turn-based ordering. A turn-based ordering is defined as follows,

\begin{definition}
\textbf{Turn-based Orderings ($\sigma$): } A turn-based ordering is a permutation $\sigma:\mathcal{C}^t_{\mathcal{O}} \rightarrow [1,k]$ over the set $\mathcal{C}^t_{\mathcal{O}}$. For any $i,j \in [1,k]$, $\sigma(a_i) = j$, equivalently $\sigma_i = j$, indicates that $a_i$ will move on the $j^\textrm{th}$ turn.
\label{def: turn_based_ordering}
\end{definition}

\noindent For a given conflict $\langle \mathcal{C}^t_{\mathcal{O}}, \mathcal{O}, t \rangle$, clearly, different permutations results in different configurations $S^{t+1}$ in $\Gamma$. Therefore, we have $\sigma \vdash \Gamma$, where $X \vdash Y$ denotes $Y$ can be obtained from the set of statements in $X$.

The conventional wisdom dictates choosing $\sigma$ randomly if agents arrive at the conflict at the same time or execute first-in first-out if they arrive asynchronously. And while $\sigma \vdash \Gamma_\textsc{opt}$ for any $\sigma$ in classical MAPF, randomly scheduling agents to move through $\mathcal{O}$ is sub-optimal, as we show in Section~\ref{sec: experiments}. The reader may further recall example $1$--opting for random ordering allows the ambulance to be delayed. In \model, we seek a unique optimal ordering, $\sigma_{\textsc{opt}}\vdash \Gamma_\textsc{opt}$, where $\Gamma_\textsc{opt}$ simultaneously minimizes Equation~\ref{eq: SoC}  as well as maximizes Equation~\ref{eq: agent_utility} for all $i$.



\section{Proposed Algorithm to Solve \model}
\label{sec: approach}

\subsection{Background on Auction Theory}

Auctions are a game-theoretic mechanism that are used extensively in economic applications like online advertising~\cite{roughgarden2016twenty}. In an auction, there are $m$ items to be allocated among $k$ agents. Each agent $a_i$ has a private valuation $v_i$ and submits a bid $b_i$ to receive at most one item dictated by an allocation rule $g_i$. A strategy is defined as an $n$ dimensional vector, $b = (b_i \cup b_{-i})$, representing the bids made by every agent. $b_{-i}$ denotes the bids made by all agents except $a_i$. The quasi-linear utility $u_i$ incurred by $a_i$ is given as follows,

\begin{equation}
 \Phi_i (b) =  v_i g_i(b) - h_i
 \label{eq: utility_template}
\end{equation}

\noindent In the equation above, the quantity on the left represents the total utility for $a_i$ which is equal to gain value of the allocated goods $v_ig_i(b)$ minus a payment term $h_i$. We refer the reader to Chapter $3$ in~\cite{roughgarden2016twenty} for a derivation and detailed analysis of Equation~\ref{eq: utility_template}. The performance of an auction is measured by the social welfare of the entire system comprising of the $k$ agents, which is defined as
\begin{equation}
\mathcal{W}(b) = \sum_{i=1}^{k} v_i g_i(b)
\label{eq: welfare_template}
\end{equation}

The primary objective in auction theory is to determine an allocation rule $g$ and a payment rule $h$ such that there exist $b^*$ for which both $u_i(b)$ for all $i$ and $\mathcal{W}(b)$ are maximized. Unfortunately, simply establishing existence is insufficient; we should be able to compute or determine $b^*$. A strategy-proof or incentive-compatible auction yields $b_i^* = v_i$.

\subsection{Algorithm}

We leverage the global optimality (Equation~\ref{eq: SoC}) of search-based MAPF solutions such as CBS and, in this section, extend their core functionality to include local agent-level optimality (Equation~\ref{eq: agent_utility}).

Formally, we identify two phases of a \model solver. The top level phase is the motion planning phase where agents make progress by stepping towards their goal states along a pre-computed trajectory cost map. The bottom level phase is the conflict resolution phase which resolves conflicts over shared goal states. Since search-based methods are susceptible to uncertainty in the edge costs~\cite{chung2019risk} and, as we show later in Section~\ref{subsec: exp-1}, result in collisions and increased time-to-goal, we rely on variants of artificial potential field (APF) methods to step towards the goal. The conflict resolution phase employs an auction to determine an optimal priority ordering in which agents should pass through the conflicted states. We describe these two stages below:
\subsubsection{Phase $1$: The motion planning phase}

We begin by creating a potential map $\mathcal{M}_c$ corresponding to each goal state $c$. Each tile (indicated by $i,j$ denoting the $i^\textrm{th}$ row and $j^\textrm{th}$ column) in $\mathcal{M}_c$ is assigned a potential value using the $A^*$ algorithm that encodes the number of tiles between the $\mathcal{G}[i,j]$ and the goal $c$. Although this formulation is designed for any number of goals, we assume $\lvert c \rvert = 1$ for simplicity.

In this work, we consider the simple one step look-ahead approach where the motion planner only plans ahead for one time-step (but may plan over multiple tiles in space). For an agent at any time step $t$ and state $s$, the motion planner steps towards a tile $\mathcal{O}$ that has a lower potential value. We assume agents always move at every time step (unless they are forced to wait during the conflict resolution phase). If there is a conflict on some $\mathcal{O}$, then the motion planner first checks if either agent can be assigned a different target state. If so, one of the agents is reassigned and the other agent is allocated the original target tile. If, however, neither agent can be reassigned, then we enter the conflict resolution phase.

The main advantage of using APF techniques is that $(i)$ they escape the exponential complexity incurred by search-based approaches and $(ii)$ we can avoid dealing with uncertainty in phase $1$ by relegating the entire responsibility of handling uncertainty to phase $2$. 




\subsubsection{Phase $2$: The conflict resolution phase}

During a conflict $\langle \mathcal{C}^t_{\mathcal{O}}, \mathcal{O}, t \rangle$, suppose $\mathcal{C}^t_{\mathcal{O}} = \{ a_1, a_2 \}$. Then, by design, either $u^t_1 \in \{ \texttt{up, down, left, right}\}, u^t_2 = \texttt{wait}$ or $u^t_2 \in \{ \texttt{up, down, left, right}\}, u^t_1 = \texttt{wait}$. A particular alternative yields a specific turn-based ordering $\sigma$. In the above example, $\sigma = \sigma_1\sigma_2$ where $\sigma_1=1,\sigma_2=2$ or $\sigma_1=2,\sigma_2=1$ Crucially, with cooperative agents as in classical MAPF, every value of $\sigma$ results in the the same outcome, but this is not so in \model, where we turn to mechanism design in order to determine $\sigma_\textsc{opt}$.

More specifically, we run an auction, $(g,h)$, with a given allocation rule $g$ and payment rule $h$. Since we are operating in simulation, we assume agents have some form of ``digital currency''. Agents $a_i$ bid on the values of $\sigma_i$ with the following allocation rule $g$:

\begin{enumerate}
    \item Given set $\mathcal{C}^t_{\mathcal{O}}$ of agents conflicted over state $\mathcal{O}$ at time $t$ and their corresponding bids $b_i$, initialize $q \gets 0$.  \label{algo: Algorithm}
    \item Sort $\mathcal{C}^t_{\mathcal{O}}$ in decreasing order of $b_i$.
    \item Do the following for $\lvert \mathcal{C}^t_{\mathcal{O}} \rvert$ steps:
    \begin{enumerate}
        \item Increment $q$ by $1$.
        \item Let $a_i$ be the the first element in $\mathcal{C}^t_{\mathcal{O}}$.
        \item Set $\sigma_i = q$.
        \item $\mathcal{C}^t_{\mathcal{O}} \gets \mathcal{C}^t_{\mathcal{O}} \setminus \{a_i\}$.
    \end{enumerate}
    \item Repeat steps $1-3$ while $\mathcal{C}^t_{\mathcal{O}}$ is non-empty. 
    \item Return 
    \begin{equation}
        \sigma_\model = \sigma_1\sigma_2 \ldots \sigma_{\lvert \mathcal{C}^t_{\mathcal{O}} \rvert}
        \label{eq: g}
    \end{equation}
\end{enumerate}

\noindent To summarize the algorithm, the agent with the highest bid is allocated the highest priority and is allowed to move first, followed by the second-highest bid, and so on. Upon moving on the $q^\textrm{th}$ turn, the agent receives a reward of $\alpha_q$ and makes a payment $h_i$ according to the following payment rule,
\begin{figure*}[t]
\centering
\begin{subfigure}[h]{0.325\textwidth}
    \includegraphics[width=\textwidth]{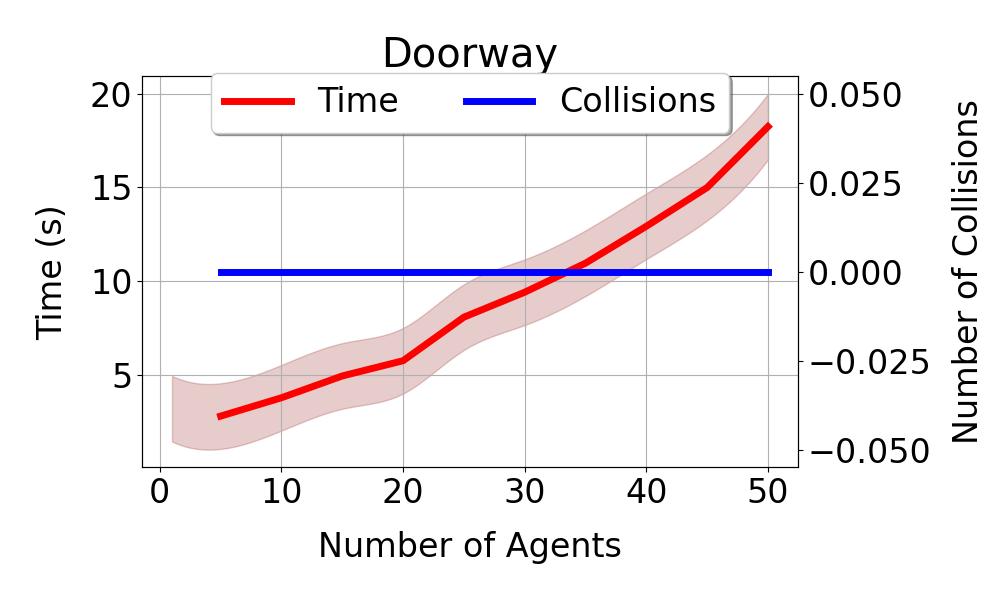}
    \caption{}
    \label{fig: doorway_socialmapf}
  \end{subfigure}
\begin{subfigure}[h]{0.325\textwidth}
    \includegraphics[width=\textwidth]{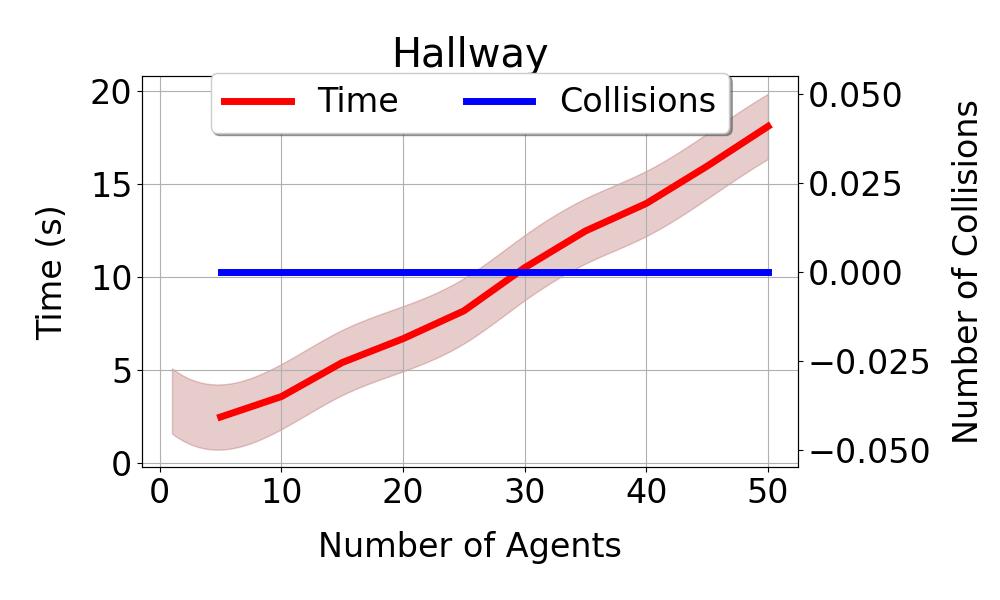}
    \caption{}
    \label{fig: hallway_socialmapf}
  \end{subfigure}
\begin{subfigure}[h]{0.325\textwidth}
    \includegraphics[width=\textwidth]{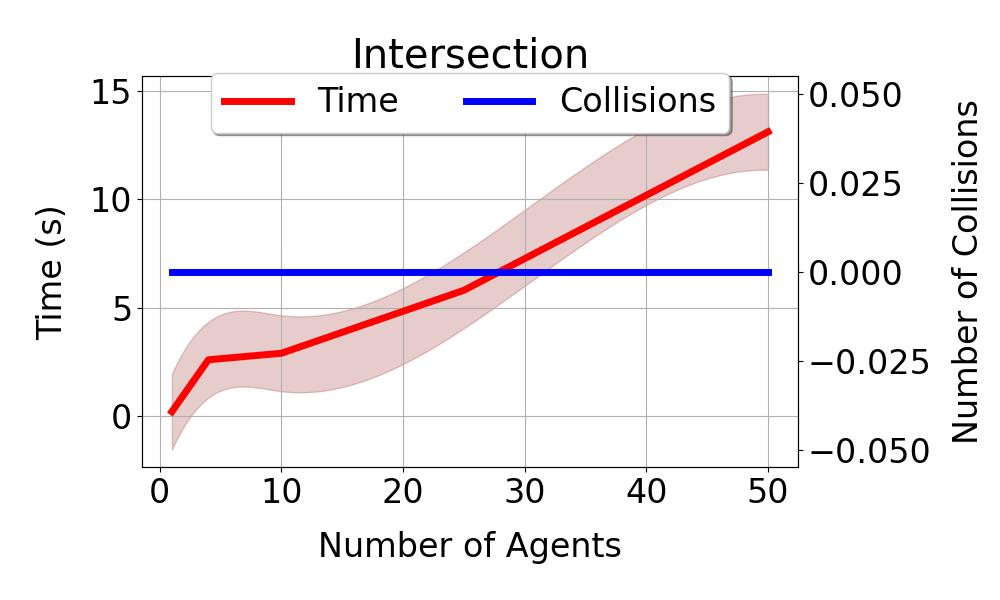}
    \caption{}
    \label{fig: intersection_socialmapf}
  \end{subfigure}
  \caption{\textbf{Auction-based planning in \model:} We measure the runtime and number of collisions with respect to the number of agents. Auction-based planners take up $\frac{1}{4}\times$ the time for $10\times$ the number of agents compared to MAPF solvers in \model settings.}
  \label{fig: MD_socialMAPF}
     \vspace{-7pt}
  \end{figure*}  
  

\begin{equation}
        h_i(b) =  \sum_{j=q}^{\lvert \mathcal{C}^t_{\mathcal{O}} \rvert} b_{j+1} \left( \alpha_j - \alpha_{j+1} \right)
     \label{eq: h}
\end{equation}

\noindent The payment rule is the ``social cost'' of reaching the goal ahead of the agents arriving on turns $q+1, q+2, \ldots, q+\mathcal{C}^t_{\mathcal{O}}$. The overall utility function for $a_i$, $\Phi_i (b)$, is given by the following equation,

\begin{equation}
        \Phi_i (b) =  v_i \alpha_q - \sum_{j=q}^{\lvert \mathcal{C}^t_{\mathcal{O}} \rvert} b_{j+1} \left( \alpha_j - \alpha_{j+1} \right)
     \label{eq: utility_social}
\end{equation}

\noindent where $b = [b_1, b_2, \ldots, b_{\mathcal{C}^t_{\mathcal{O}}}]^\top$ represent the agent bids. The gain term $v_i \alpha_q$ denotes $a_i$'s time reward on reaching the goal on the $q^\textrm{th}$ turn. Note that the gain term in Equation~\ref{eq: utility_social} $(v_i \alpha_q)$ implicitly encourages agents to adopt shorter paths (greater time rewards) but not necessarily the shortest path. Such a formulation recognizes that social navigation is not equivalent to traversing the shortest paths, as typically formulated in the classical MAPF.

To obtain the global objective function for a sequence $\Gamma$, we first reformulate Equation~\ref{eq: SoC} as

\begin{equation}
    \mathcal{\widehat F}(\Gamma) = \sum_{i=1}^{k} v_i t_{g,i}
    \label{eq: socialmapf_objective}
\end{equation}

\noindent Minimizing Equation~\ref{eq: socialmapf_objective} is the same as minimizing the SoC since the vector $[v_1, v_2, \ldots, v_k]^\top$ is constant. Equivalently, using the A.M.-G.M.-H.M. inequality, we can choose to maximize $\sum_{i=1}^{k} {v_i}(t_{g,i})^{-1}$. Using $\alpha_i = (t_{g,i})^{-1}$, we have,

\begin{equation}
    \begin{split}
    \Gamma^* = \arg_{\Gamma}\min \mathcal{\widehat F}(\Gamma)
    &= \arg_{\Gamma}\max \sum_{i=1}^{k} v_i (t_{g,i})^{-1}\\
    &= \arg_{\Gamma}\max \sum_{i=1}^{k} v_i \alpha_i  \\
\end{split}
\label{eq: socialmapf_objective_welfare}
\end{equation}
\noindent which represents the social welfare, $\mathcal{W}$, of an auction with $k$ agents that receives rewards $\alpha_i$. 

We have, at this point, derived the existence of $\sigma_\model \vdash \Gamma$ with Equations~\ref{eq: utility_social} and \ref{eq: socialmapf_objective_welfare} corresponding to the agent utility function and social welfare of a sequence $\Gamma$. We now prove $\Gamma$ is optimal. 

\begin{theorem}
\textbf{$\sigma_\model \vdash \Gamma_\textsc{opt}$:} If $\sigma_\model \gets (g,h)$ where $(g,h)$ are defined by Equations~\ref{eq: g} and \ref{eq: h}, then $\sigma_\model \vdash \Gamma_\textsc{opt}$.
\label{thm: optimality}
\end{theorem}

\begin{proof}
From Definition~\ref{def: optimal_gamma}, $\Gamma$ is optimal when the global sum-of-cost (Equation~\ref{eq: SoC}) is minimal and local agent utilities (Equation~\ref{eq: agent_utility}) are maximal. Equations~\ref{eq: utility_social} and \ref{eq: socialmapf_objective_welfare} demonstrate that Equations~\ref{eq: SoC} and \ref{eq: agent_utility} correspond to the utility and welfare functions of an arbitrary auction. Therefore, to prove that $\sigma_\model \vdash \Gamma_\textsc{opt}$, it is sufficient to show that $(g,h)$ is welfare-maximizing and strategy-proof. Theoretical analysis of the auction $(g,h)$ in~\cite{chandra2022gameplan, suriyarachchi2022gameopt} showed that the auction specified by $(g,h)$ in Equation~\ref{eq: utility_social} is a strategy-proof, welfare-maximizing auction with the following strictly dominant strategy: $b_i^* = v_i$.
\end{proof}



\section{Experiments and Discussion}
\label{sec: experiments}

Our experiments answer the following questions: $(i)$ how does mechanism design-based planning perform in \model? $(ii)$ is mechanism design-based planning strategy-proof? $(iii)$, can this new class of mechanism design-based solvers extend to increasingly complex \model environments, for \textit{e.g.} in the presence of obstacles? and  $(iv)$ can current optimal search-based MAPF algorithms solve \model?

We used Python to perform all experiments on an Intel(R) i$7$ CPU at $2.20$GHz with $8$ cores and $16$GB RAM. For experiments that involve comparing runtime, we set a maximum timeout for each trial as $20$ seconds.
\begin{figure*}[t]
\begin{subfigure}[h]{0.245\textwidth}
    \includegraphics[width=\textwidth]{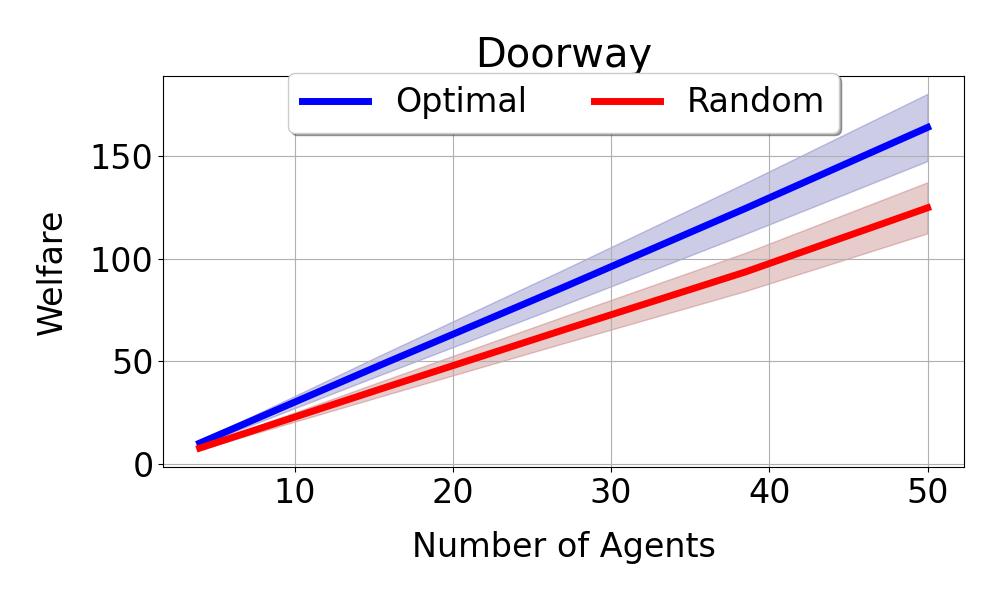}
    \caption{}
    \label{fig: social_welfare_doorway}
  \end{subfigure}
\begin{subfigure}[h]{0.245\textwidth}
    \includegraphics[width=\textwidth]{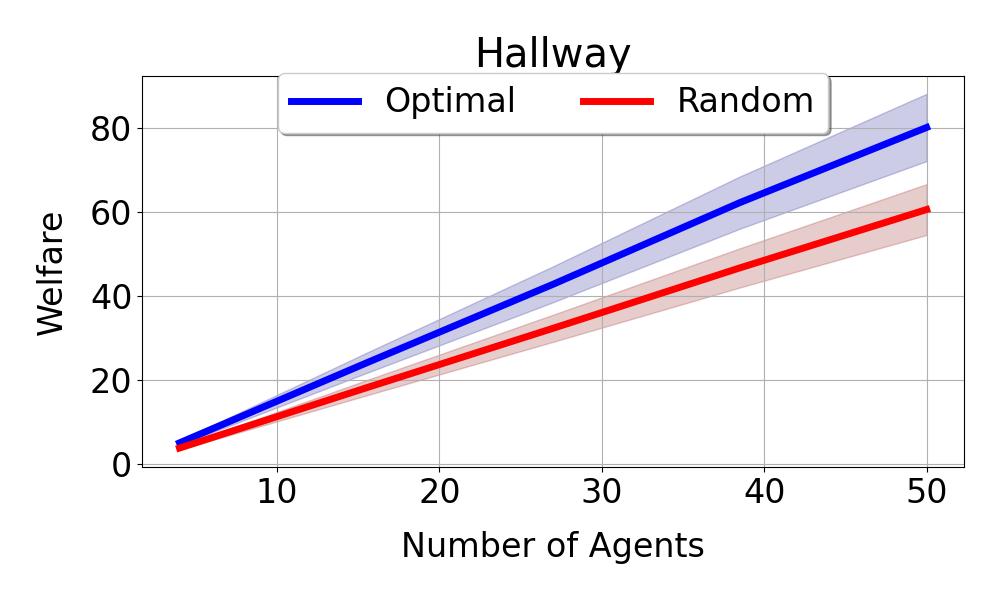}
    \caption{}
    \label{fig: social_welfare_hallway}
  \end{subfigure}
\begin{subfigure}[h]{0.245\textwidth}
    \includegraphics[width=\textwidth]{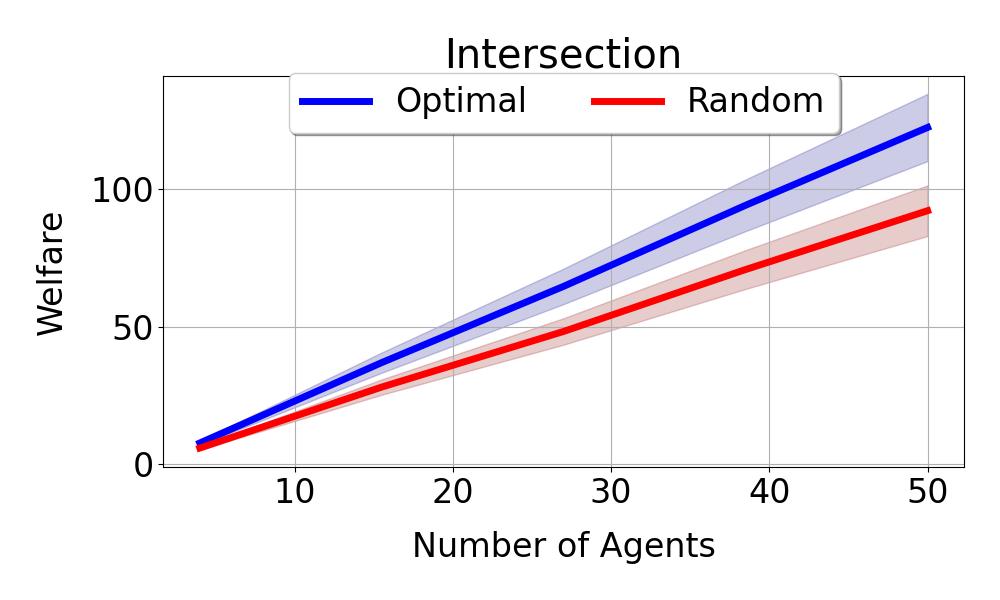}
    \caption{}
    \label{fig: social_welfare_intersection}
  \end{subfigure}
\begin{subfigure}[h]{0.245\textwidth}
    \includegraphics[width=\textwidth]{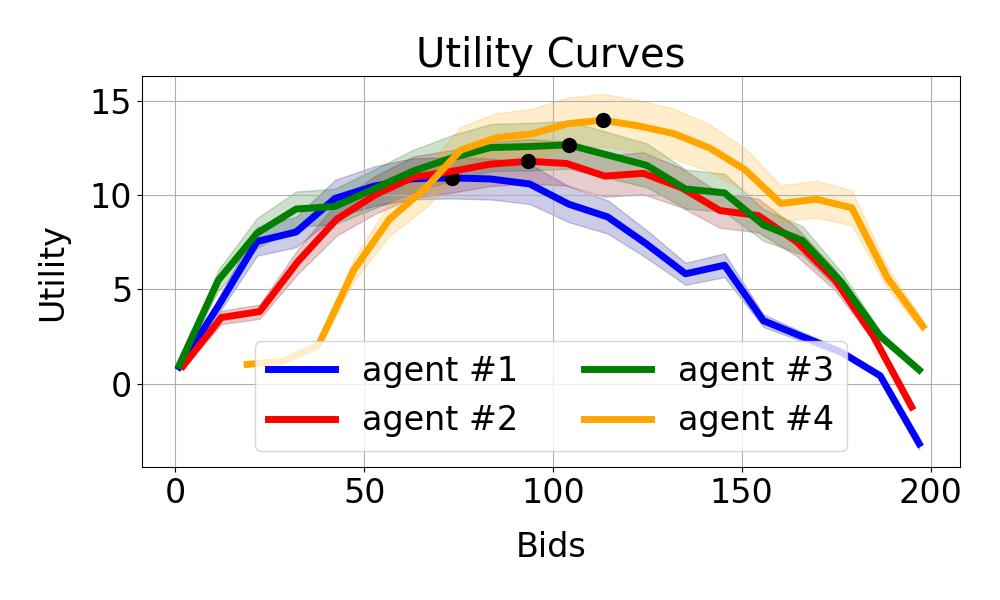}
    \caption{}
    \label{fig: utility}
  \end{subfigure}
\caption{\textbf{Auction-based planning is strategy-proof: }In Figures~\ref{fig: social_welfare_doorway},~\ref{fig: social_welfare_hallway}, and~\ref{fig: social_welfare_intersection}, we compare the social welfare of the system using an auction planner (blue line) with that computed using CBS-random (red line). In Figure~\ref{fig: utility}, we plot the utility curves for $4$ agents as a function of their bids. The black points denote an agent's private incentive $v_i$.}
  \label{fig: strategy_proof}
\end{figure*}

\begin{figure*}[t]
\centering
  \begin{subfigure}[h]{0.24\textwidth}
    \includegraphics[width=\textwidth]{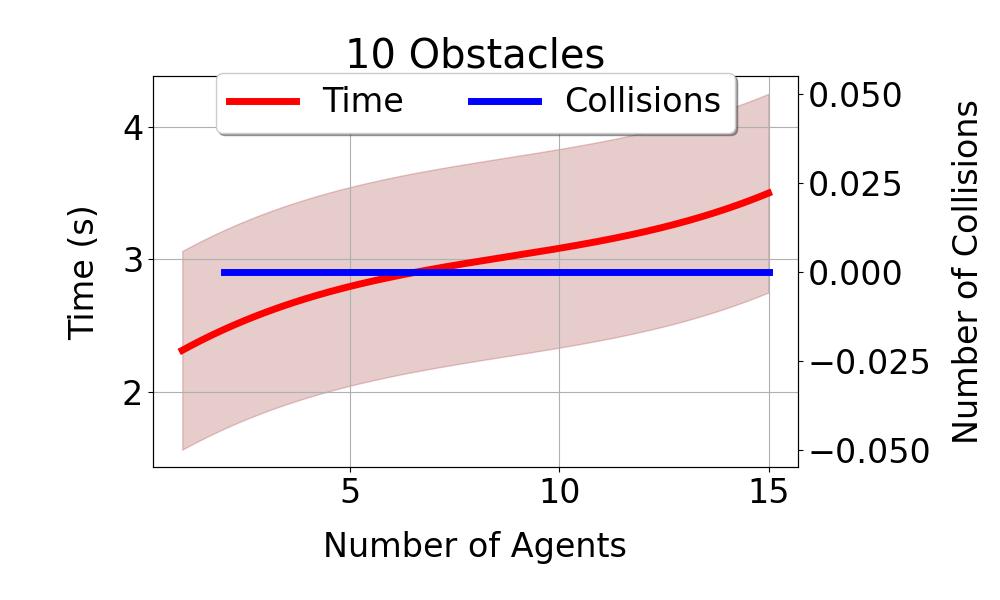}
    \caption{}
    \label{fig: obs1}
  \end{subfigure}
  \begin{subfigure}[h]{0.24\textwidth}
    \includegraphics[width=\textwidth]{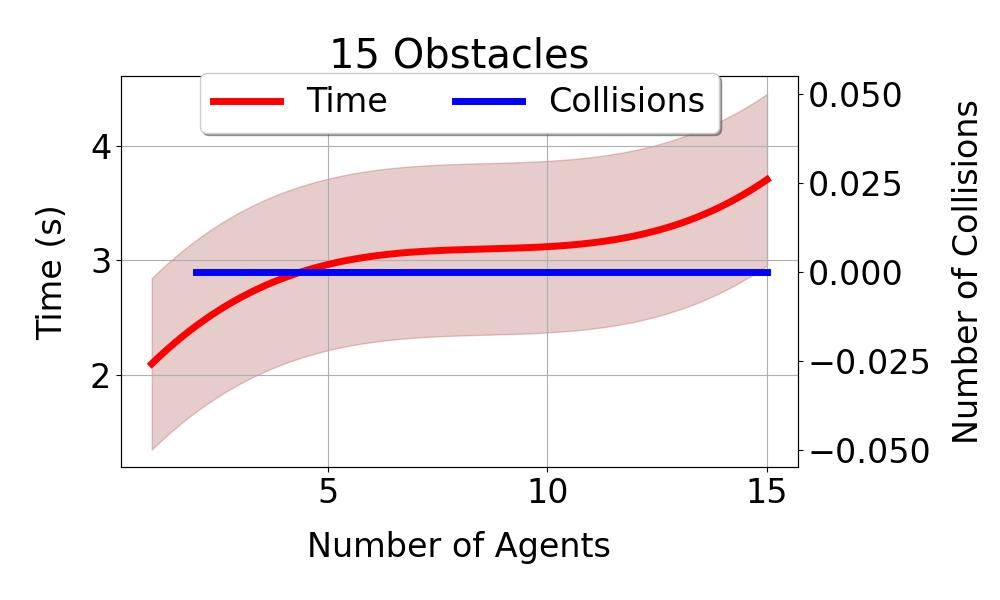}
    \caption{}
    \label{fig: obs2}
  \end{subfigure}
  \begin{subfigure}[h]{0.24\textwidth}
    \includegraphics[width=\textwidth]{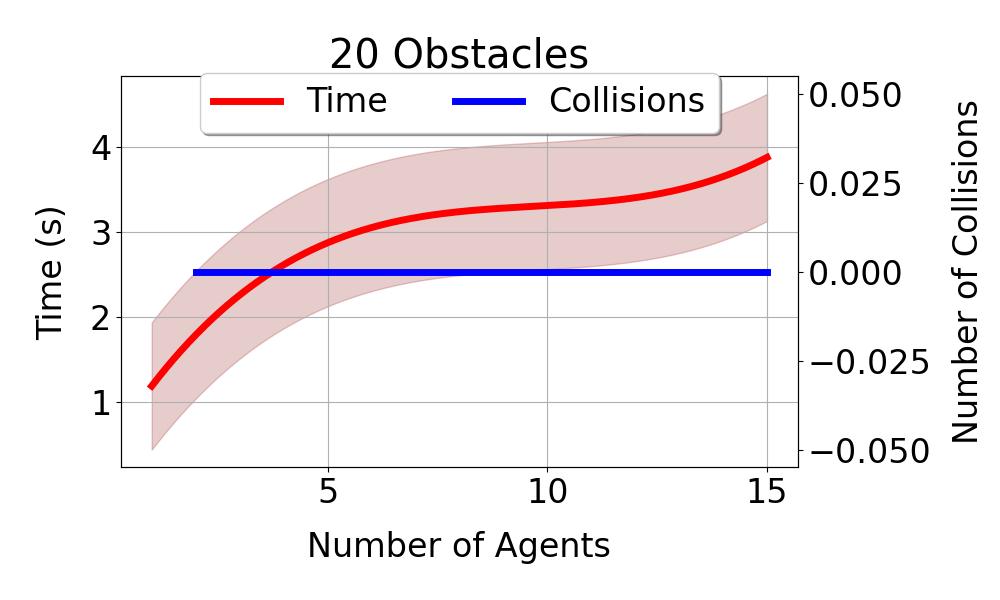}
    \caption{}
    \label{fig: obs3}
  \end{subfigure}
  \begin{subfigure}[h]{0.24\textwidth}
    \includegraphics[width=\textwidth]{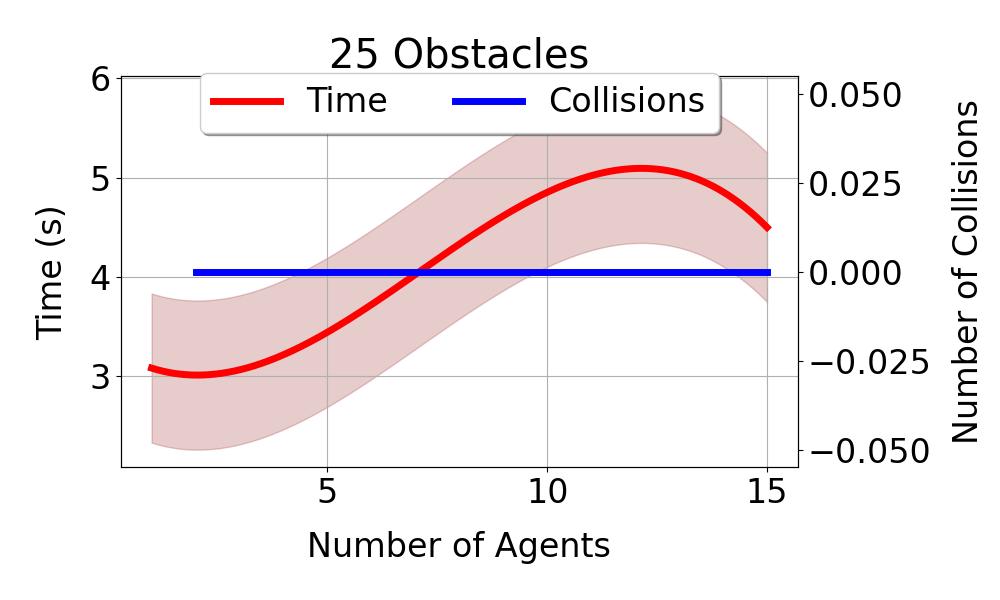}
    \caption{}
    \label{fig: obs4}
  \end{subfigure}

\caption{\textbf{\model with obstacles:} Measuring the time to goal and number of collisions with respect to number of obstacles randomly placed in a fixed grid size.}
  \label{fig: obstacles}
\end{figure*}   

\begin{figure*}[t]
\centering
   \begin{subfigure}[h]{0.19\textwidth}
    \includegraphics[width=\textwidth]{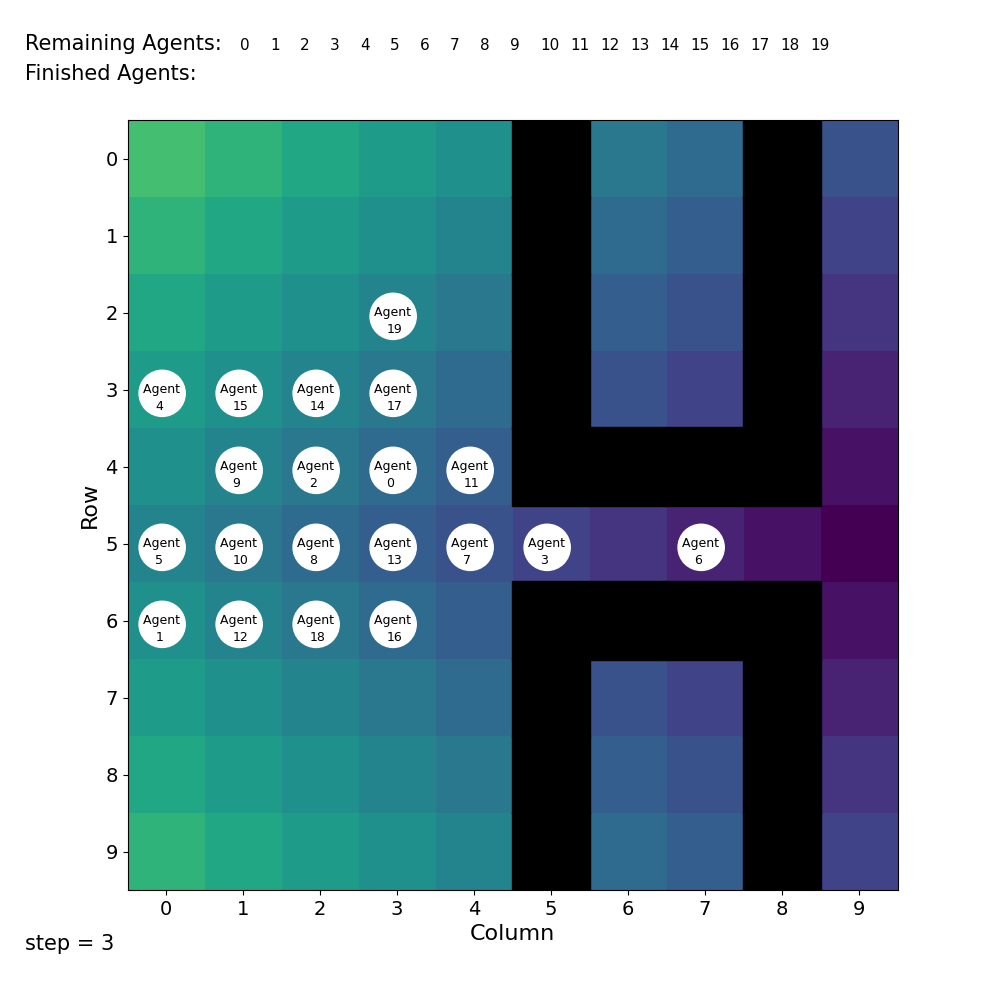}
    \caption{}
    \label{fig: doorway_agents}
  \end{subfigure}
 \begin{subfigure}[h]{0.19\textwidth}
    \includegraphics[width=\textwidth]{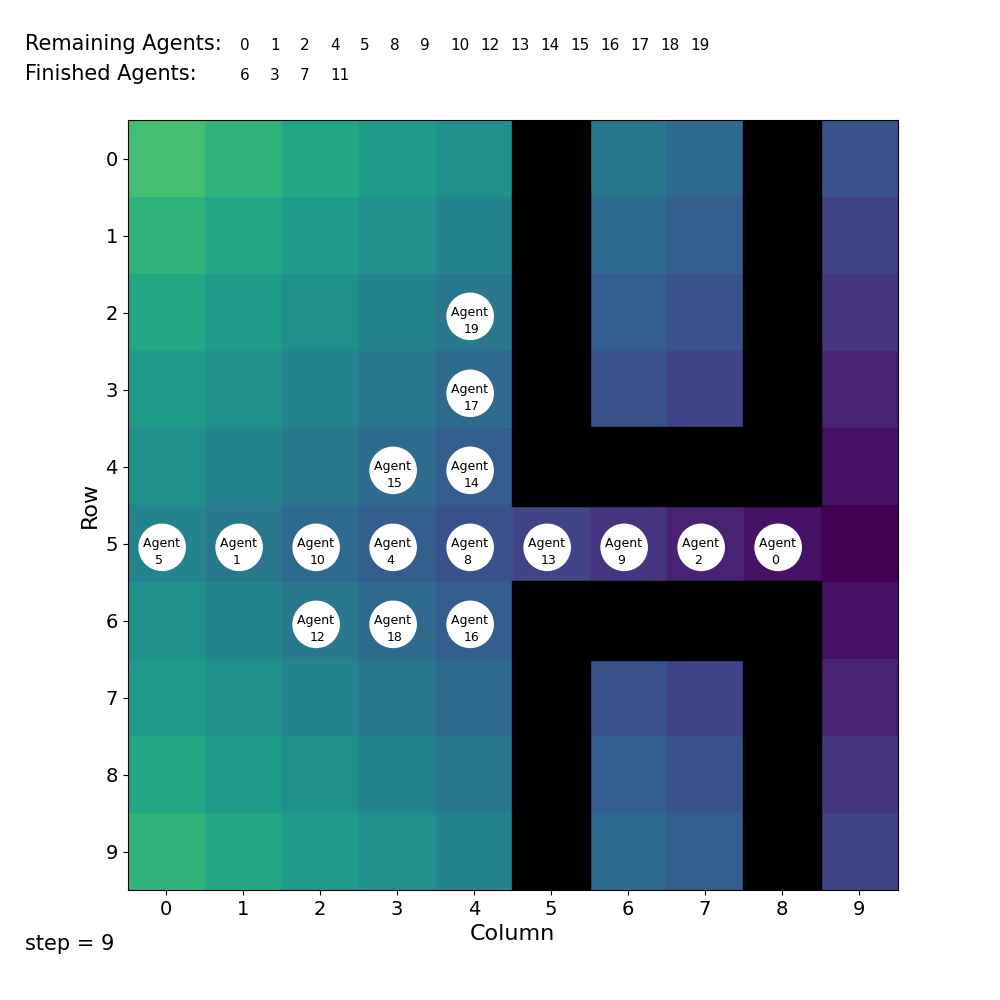}
    \caption{}
    \label{fig: hallway_agents}
  \end{subfigure}
\begin{subfigure}[h]{0.19\textwidth}
    \includegraphics[width=\textwidth]{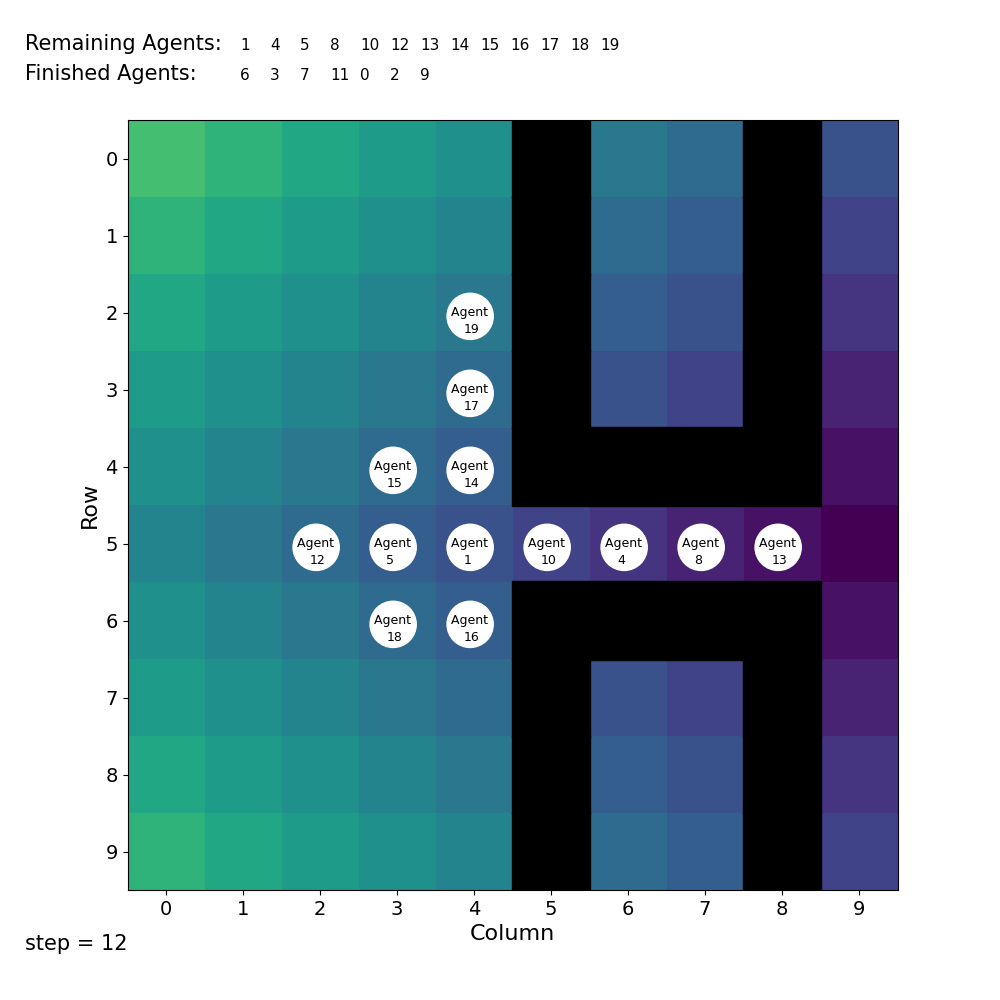}
    \caption{}
    \label{fig: intersection_agents}
  \end{subfigure}
\begin{subfigure}[h]{0.19\textwidth}
    \includegraphics[width=\textwidth]{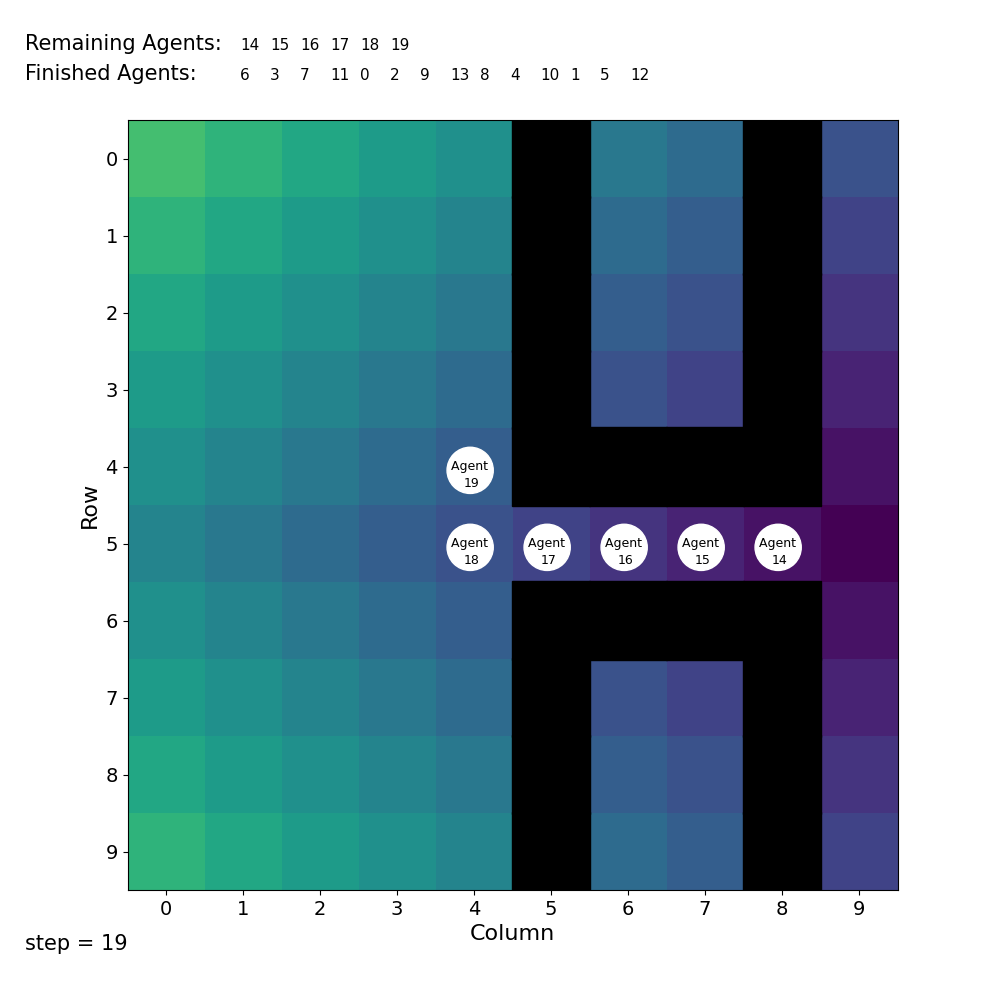}
    \caption{}
    \label{fig: doorway_gapsize}
  \end{subfigure}
\begin{subfigure}[h]{0.19\textwidth}
    \includegraphics[width=\textwidth]{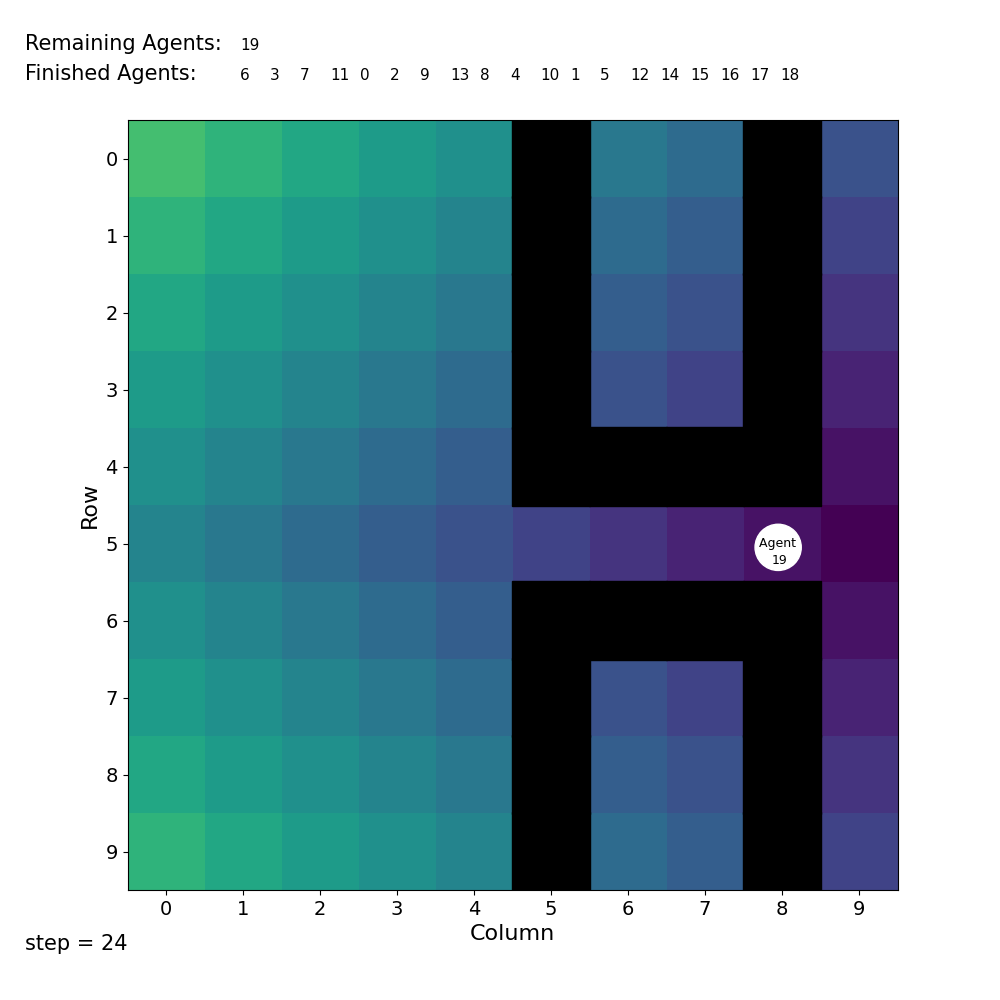}
    \caption{}
    \label{fig: hallway_gapsize}
  \end{subfigure}
   \begin{subfigure}[h]{0.19\textwidth}
    \includegraphics[width=\textwidth]{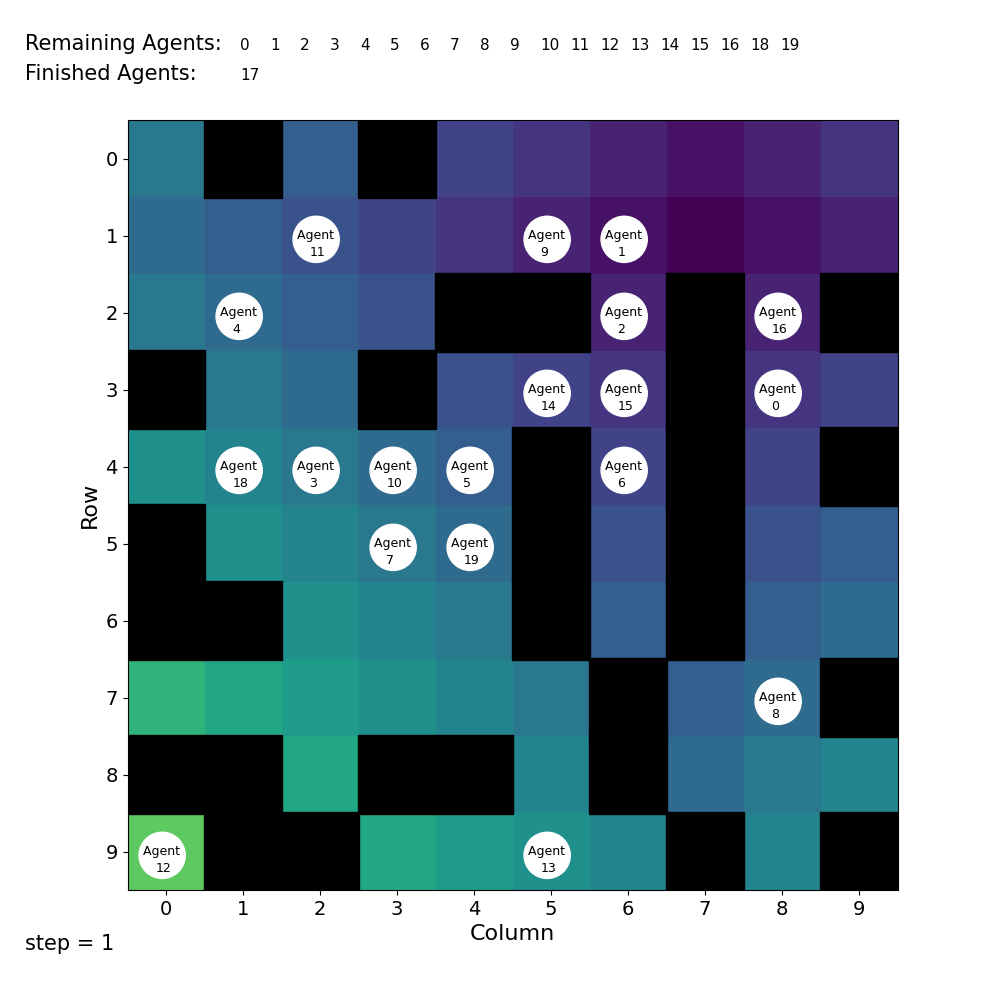}
    \caption{}
    \label{fig: obs_vis_1}
  \end{subfigure}
 \begin{subfigure}[h]{0.19\textwidth}
    \includegraphics[width=\textwidth]{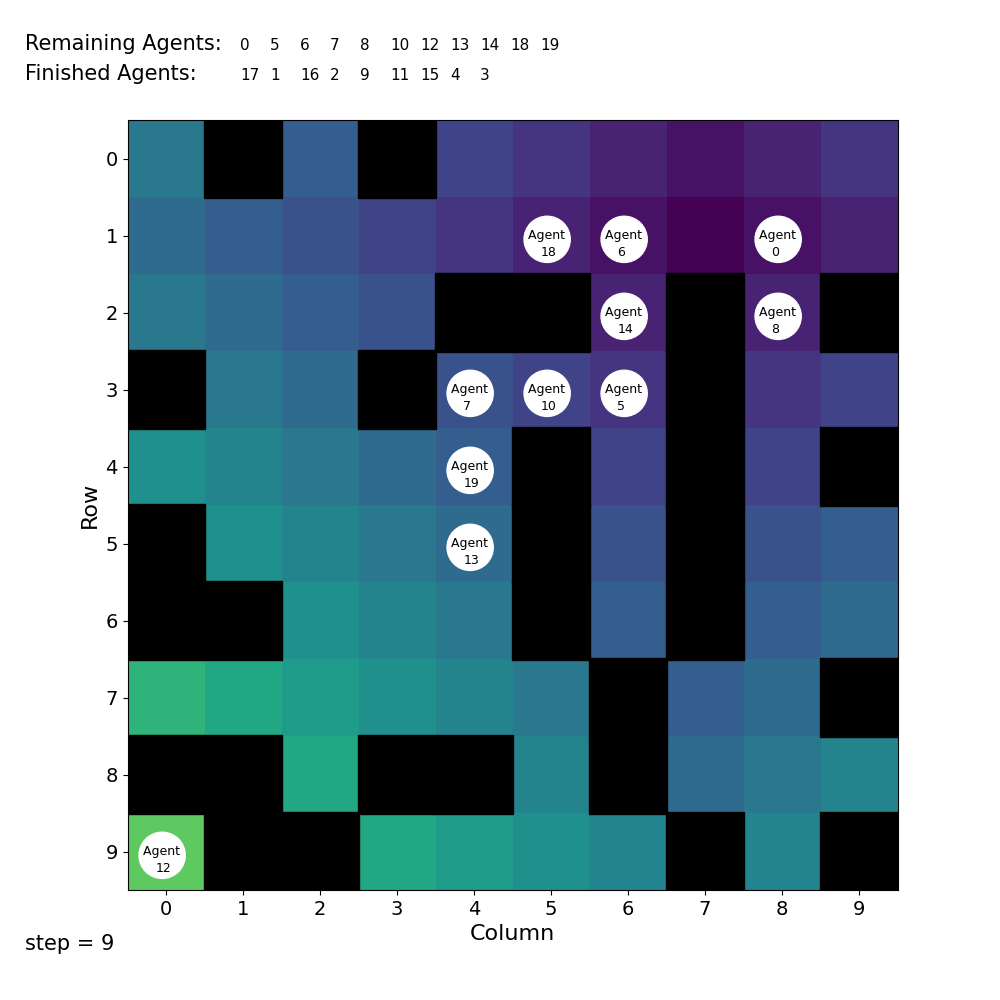}
    \caption{}
    \label{fig: obs_vis_2}
  \end{subfigure}
\begin{subfigure}[h]{0.19\textwidth}
    \includegraphics[width=\textwidth]{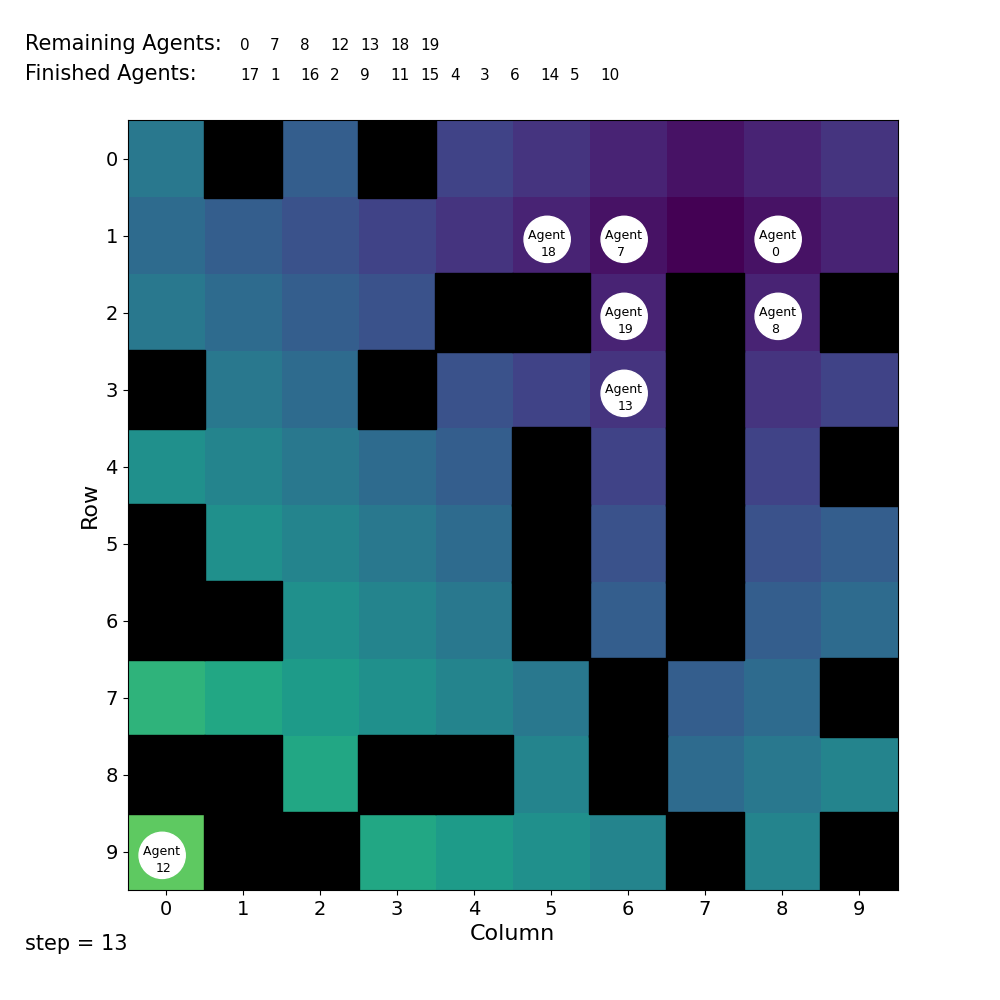}
    \caption{}
    \label{fig: obs_vis_3}
  \end{subfigure}
\begin{subfigure}[h]{0.19\textwidth}
    \includegraphics[width=\textwidth]{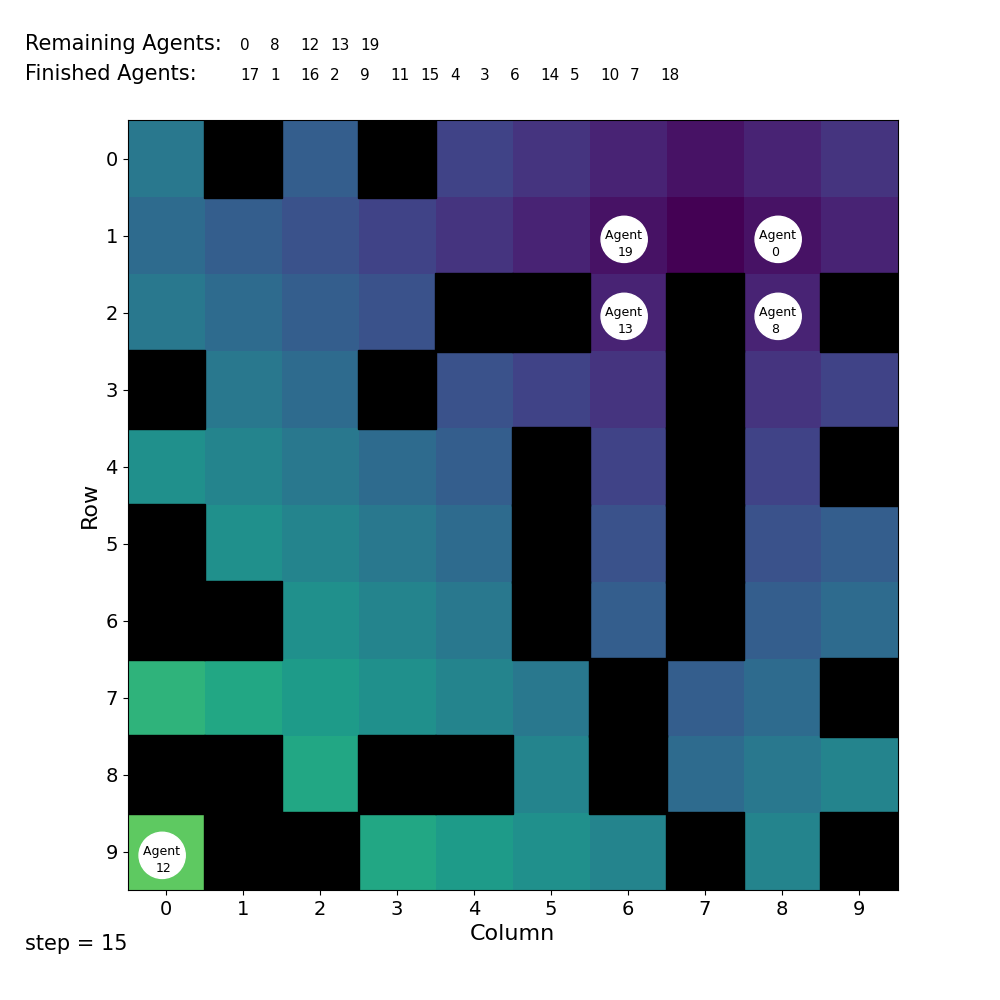}
    \caption{}
    \label{fig: obs_vis_4}
  \end{subfigure}
\begin{subfigure}[h]{0.19\textwidth}
    \includegraphics[width=\textwidth]{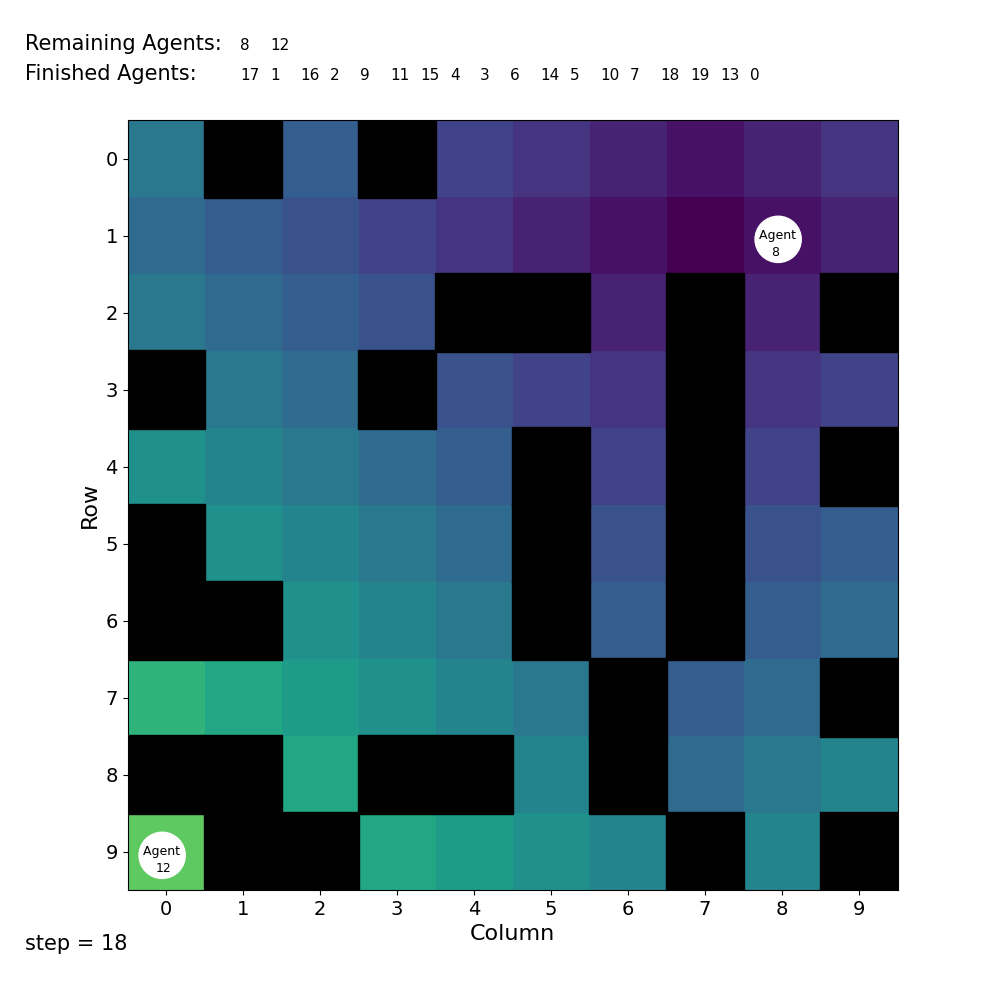}
    \caption{}
    \label{fig: obs_vis_5}
  \end{subfigure}
\caption{\textbf{Visualizations:} \textit{(top)} auction-based planning in the narrow hallway scenario for many strategic agents and \textit{(bottom)} auction-based planning in the presence of static obstacles.}
  \label{fig: simulation}
  \vspace{-10pt}
\end{figure*}

\subsection{Mechanism Design-based Planning in \model}

In contrast, we test our new auction-based \model solver in the same doorway, hallway, and intersection environments and compare the time taken to reach the goal as well as the number of collisions. Figures~\ref{fig: doorway_socialmapf}-\ref{fig: intersection_socialmapf} clearly show that our proposed approach yields $0$ collisions (whereas CBS and CBS-random yield up to $50$ and $20$ collisions, respectively, in the intersection scenario) and can scale up to $\mathcal{O}(n^2)$ agents where $n$ refers to the size of the grid. We visualize our approach in Figure~\ref{fig: simulation}. We vary the number of agents from $4$ to $50$ and average the results over $100$ trials. The solid color lines represent the mean value with the shaded regions representing the confidence intervals. We note that while the runtime is still excessive compared to modern MAPF approaches~\cite{li2021eecbs}, we emphasize that we have not implemented any performance-enhancing heuristic optimization, which is a direction of future work. The goal was to highlight the fact that in social navigation scenarios, where existing MAPF methods take over $60$ seconds to find a solution for $5$ or more agents, auction-based planners take up $\frac{1}{4}\times$ the time for $10\times$ the number of agents. 

\begin{figure*}[t]
\centering
\begin{tabular}{cccc}
    \toprule
    & Doorway & Hallway & Intersection\\
    \midrule
          \multirow{2}{*}{\rotatebox{90}{CBS}} &  \includegraphics[width=0.275\textwidth]{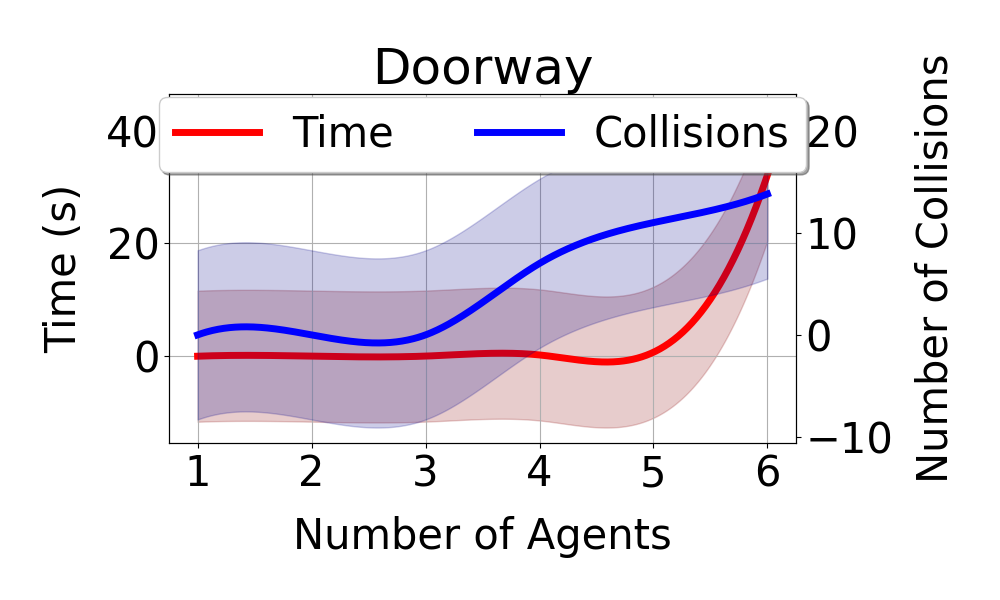} &     \includegraphics[width=0.275\textwidth]{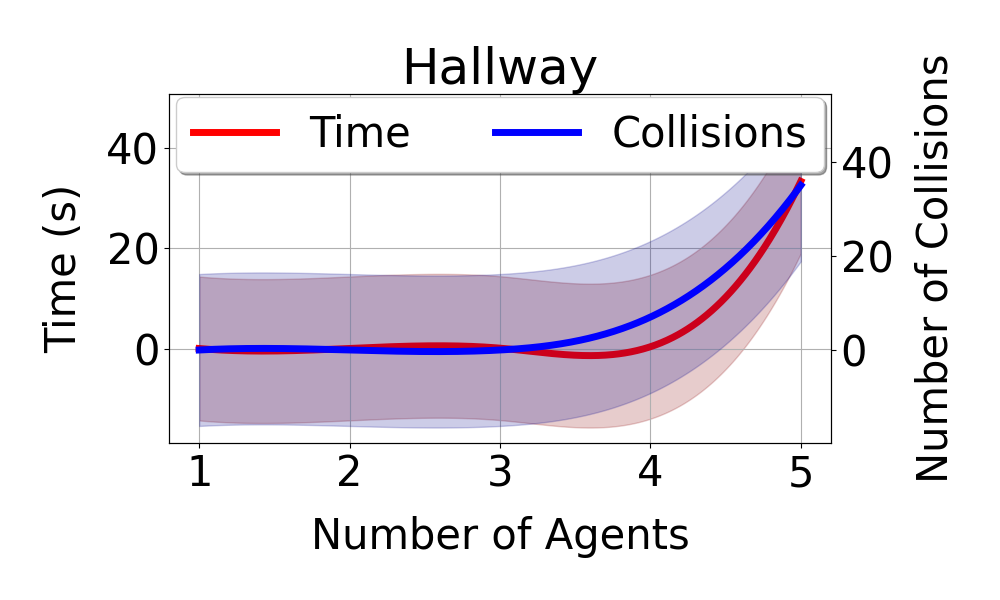}&     \includegraphics[width=0.275\textwidth]{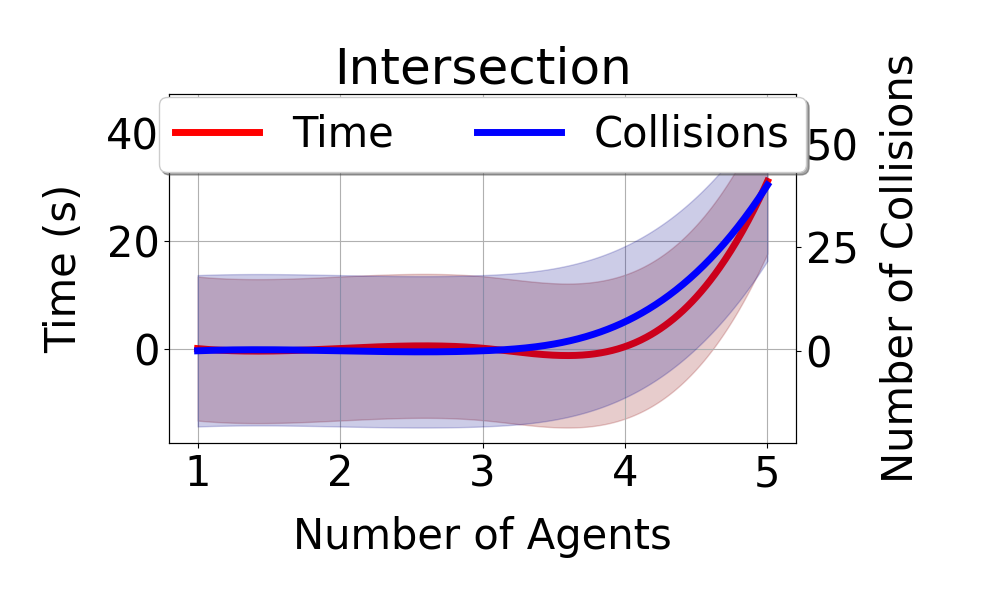} \\
     &  \includegraphics[width=0.275\textwidth]{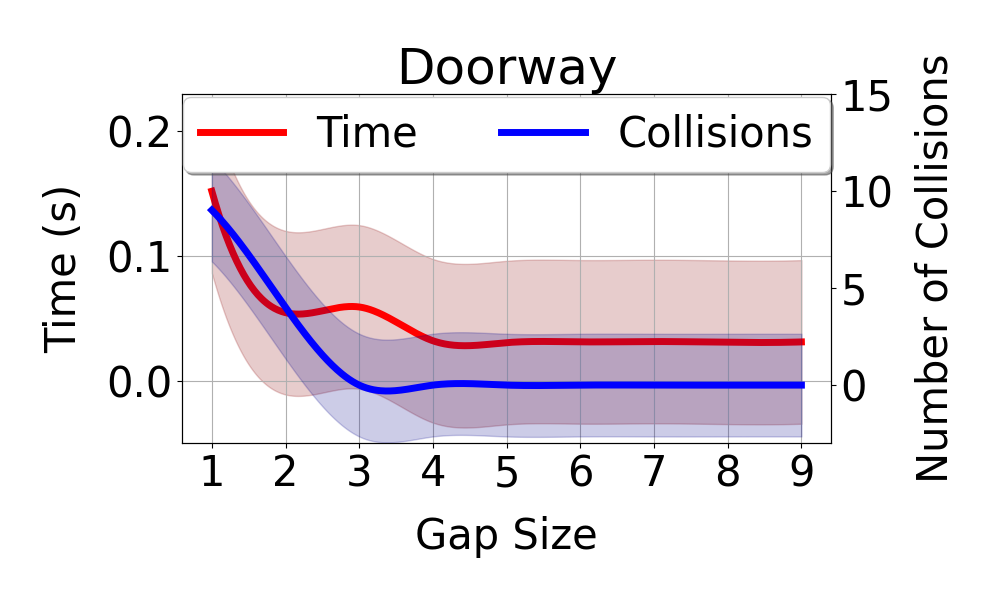} &     \includegraphics[width=0.275\textwidth]{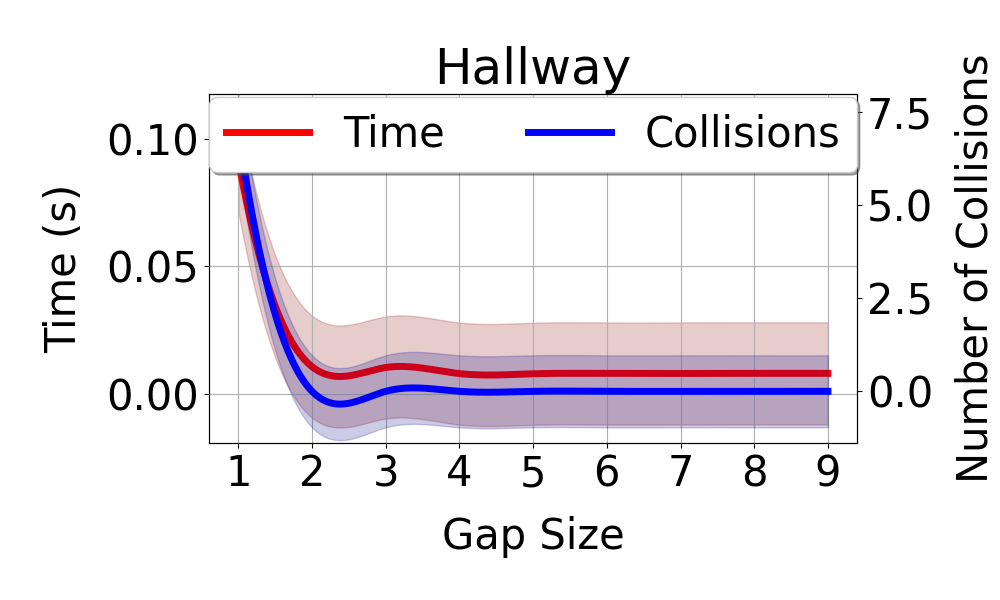}&     \includegraphics[width=0.275\textwidth]{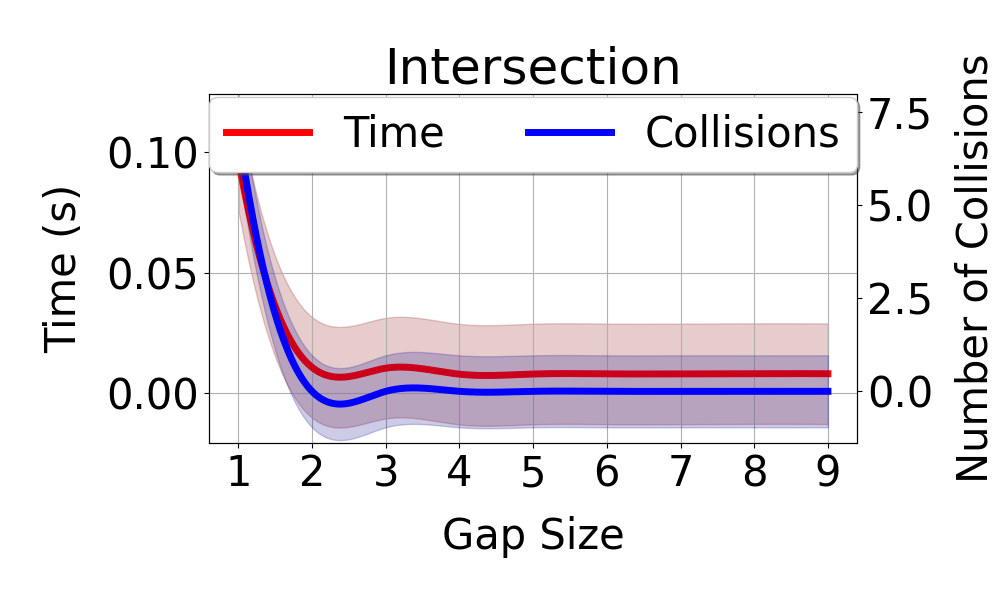} \\
      \cmidrule{2-4}
       \multirow{2}{*}{\rotatebox{90}{CBS-random}}&  \includegraphics[width=0.275\textwidth]{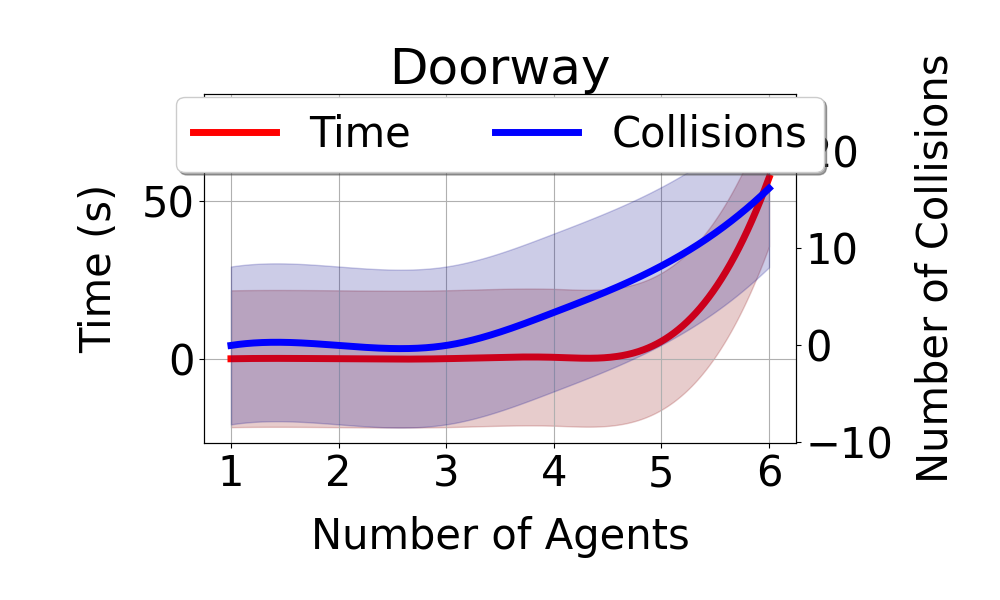} &     \includegraphics[width=0.275\textwidth]{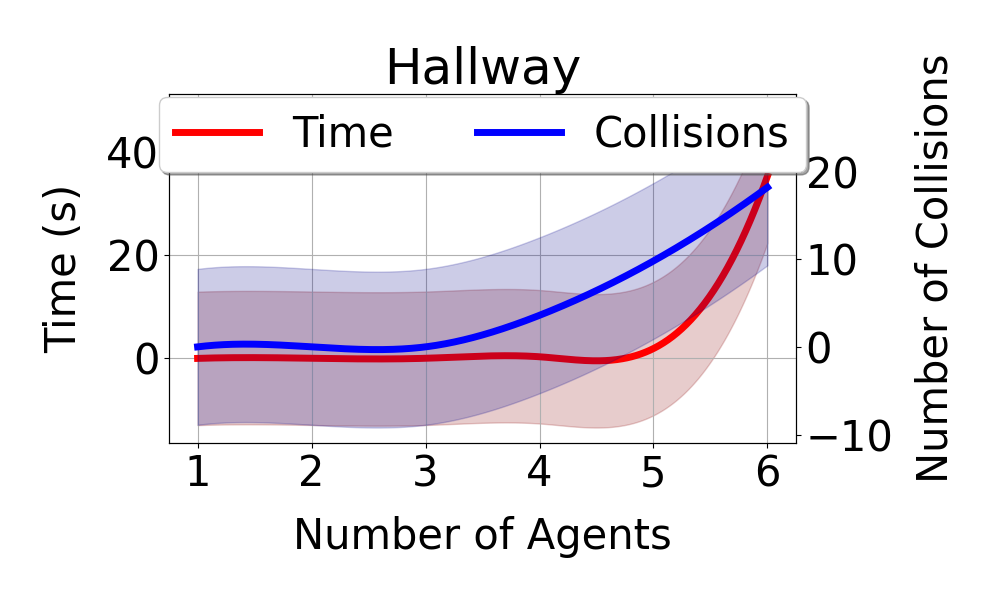}&     \includegraphics[width=0.275\textwidth]{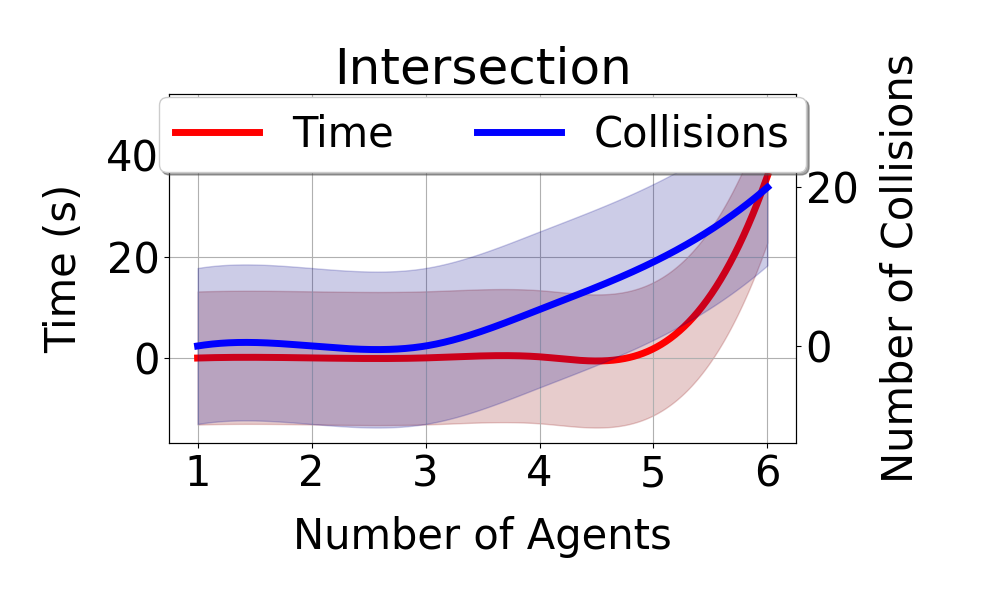} \\
      &  \includegraphics[width=0.275\textwidth]{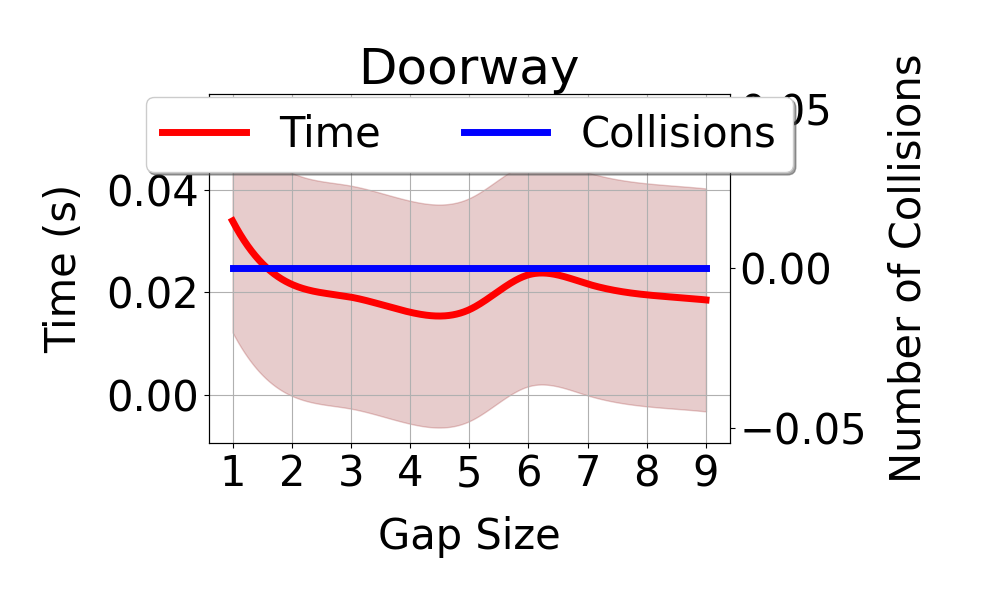} &     \includegraphics[width=0.275\textwidth]{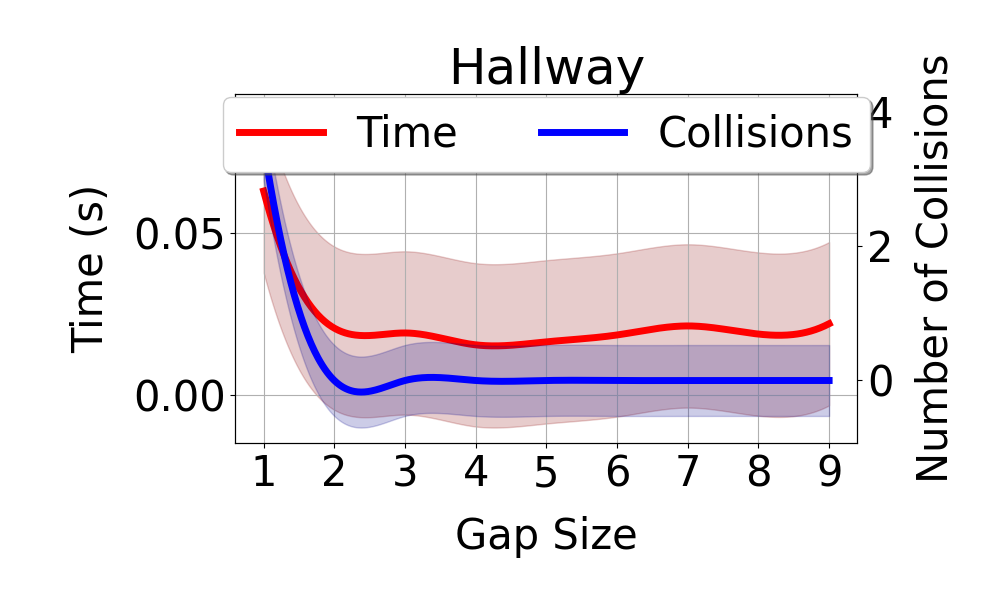}&     \includegraphics[width=0.275\textwidth]{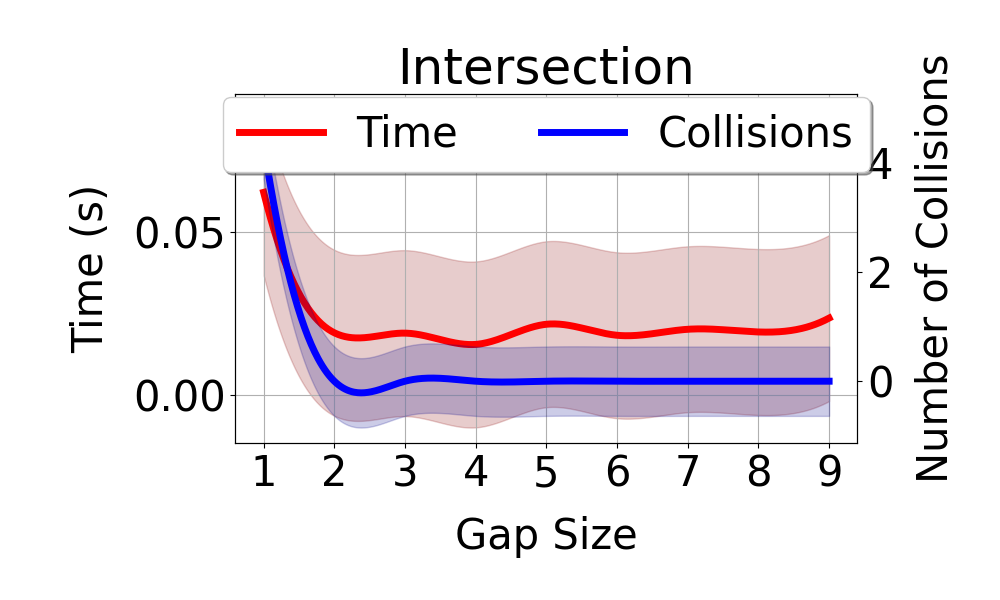} \\
     \bottomrule
\end{tabular}
\caption{\textbf{Optimal search-based planning in \model:} We analyze optimal search-based MAPF solvers such as conflict-based search (CBS)~\cite{CBS} and its variant, CBS-random~\cite{ecbs}, in the \model regime. See Section~\ref{subsec: exp-1} for a detailed analysis.}
  \label{fig: CBS_main}
  \vspace{-10pt}
\end{figure*}

\subsection{Mechanism Design-based Planning is Strategy-Proof}

One of the main challenges with MAPF with strategic agents is to simultaneously minimize the cost to the overall system-wide objective (Equation~\ref{eq: socialmapf_objective_welfare}) and maximize each agent's utility functions (Equation~\ref{eq: utility_social}). In this section, we empirically confirm that mechanism design (more specifically, auctions) satisfies both objectives.

In Figures~\ref{fig: social_welfare_doorway},~\ref{fig: social_welfare_hallway}, and~\ref{fig: social_welfare_intersection}, we compute the system welfare for auction-based planning and CBS-random for the doorway, narrow hallway, and corridor intersection scenarios. As before, we vary the number of agents from $4$ to $50$ and average the results over $100$ trials. The solid color lines represent the mean value with the shaded regions representing the confidence intervals. The blue line corresponds to the system welfare when the bidding strategy is optimally according to Theorem~\ref{thm: optimality}. The red line, on the other hand, corresponds to the system welfare when the bidding strategy is chosen randomly. The linear relationship with increasing number of agents is expected since the social welfare is a monotonically increasing linear function.  

In addition, in Figure~\ref{fig: utility}, we plot the utility curves for $4$ agents performing \model in the narrow hallway scenario as depicted in Figure~\ref{fig: demo_figure}. On the vertical axis, we plot the utility values corresponding to bids (horizontal axis) made by each agent. The black circles indicate the agents' true incentives. We observe that bidding the true incentive yields the maximum utility (highest points on each corresponding curve). Note that in simulated settings, the private incentives can be chosen arbitrarily. In real world navigation with robots, the incentives would need to be more carefully assigned and is a direction of future work.

\subsection{Extending Mechanism Design-based Planning to \model with Static Obstacles}

We test mechanism design-based planning in \model environments with static obstacles randomly placed throughout the $10\times 10$ grid as visualized in Figures~\ref{fig: obs_vis_1}-\ref{fig: obs_vis_5}. As before, we benchmark the runtime and number of collisions by averaging over $100$ trials each. In Figure~\ref{fig: obstacles}, in each graph, we vary the number of agents on the horizontal axis from $4$ to $15$ never exceeding more than $5$ seconds in each case. We vary the number of obstacles from $10$ to $25$ and observe $0$ collisions. Note that both the initial configurations of the agents as well as the obstacles are random. We observe that auction-based planning has no trouble in scaling in terms of number of obstacles and number of agents.

\subsection{Optimal Search-based algorithms in \model}
\label{subsec: exp-1}

We investigate the performance of an optimal search-based algorithm such as conflict-based search (CBS)~\cite{CBS} in \model. Since CBS is originally designed for the classical MAPF formulation with cooperative agents with full observability, we model uncertainty for the private incentives in the CBS implementation by normally distribute the edge costs in the planning stage and allowing multi-hop path traversal during the execution stage. In Figure~\ref{fig: CBS_main}, the columns correspond to the doorway, narrow hallway, and corridor intersection, respectively. We benchmark the CBS algorithm (rows $1$ and $2$) and its variant, called CBS-random (rows $3$ and $4$), where conflicted nodes with equal costs for the child nodes are selected randomly (as opposed to first-come first-serve). The first and third rows of Figure~\ref{fig: CBS_main} test these solvers as the number of strategic agents increases and the second and fourth rows test the effect of the constrained geometry of the environment (specifically, we increase the gap size of the doorway or hallway). In all experiments, we average the results over $100$ trials. The solid color lines represent the mean value and the associated shaded regions represent the confidence intervals. 

In general, we observe that CBS and CBS-random scale poorly as the number of agents increases and no longer guarantee collision-free trajectories. For more than $4$ agents, search-based methods fail to return a solution within the allotted time ($20$ seconds). Furthermore, we observe that CBS incurs up to $50$ collisions in the intersection environment. To test the effect of the social navigation environment, we vary the gap size from $1$ (narrow gap) to $9$ (open environment) and, as expected, observe that collisions are more when the environment is constrained (smaller gap sizes). Since in this experiment, we keep the number of the agents fixed at $3$, we do not observe variations in the runtime.

\section{Conclusion, Limitations, and Future Work}
\label{sec: conclusion}

We presented a new framework for multi-agent path finding (MAPF) in social navigation scenarios like doorways, elevators, hallways, and corridor intersections among agents with private incentives. We show that existing search-based MAPF solvers are unable to model these unknown incentives and result in collisions due to the inherent uncertainty in agent intentions. Furthermore, we show that these existing methods are unable to provide efficiency and optimality guarantees. Our solution consists of an auction-based approach where we resolve conflicts by incentivising agents to reveal their true intents and executing a conflict resolution protocol based on these true intents. We show that our auction planner results in zero collisions and is $20\times$ more efficient than search-based solvers. 

Our work is intended to foster new research directions in artificial intelligence and robotics while simultaneously shedding light on hidden connections with existing, and seemingly disconnected, research in machine learning and computer vision. In each domain, there are fundamental open problems and research directions that need to be solved in order to achieve actual physical robots navigating among humans in the real world. \model is expected to benefit researchers in the human-robot interaction, artificial intelligence, robotics communities working towards social robot navigation. We outline several independent research themes below that build upon the \model framework:

\begin{enumerate}
    \item \textit{Extending \model solvers to continuous space-time domains:} This work proposed the simple case of discrete time and discrete space. While \model encodes velocity constraints in the utility function in Equation~\ref{eq: utility_social}, thereby facilitating kinodynamic planning, there are rich opportunities in terms of how collisions would need to be redefined.
    
    \item \textit{Modeling uncertainty in agent state-spaces:} This paper modeled uncertainty by sampling edge costs in graph search-based MAPF solvers from a random Gaussian distribution. There exists, however, an entire field of study devoted to planning under uncertainty from which additional models of uncertainty can be used.
    
    \item \textit{Human incentive estimation:} In this work, we use simple heuristics like the velocity of the agents and their time-to-goal to determine their incentives or priority towards their goals. However, more sophisticated estimation techniques including bayesian models, neural networks, and computer vision may also be employed.
    
    \item \textit{Better motion planners:} We used artificial potential field approaches for the motion planner and leveraged heuristic strategies to escape local minima. The greedy one step look-ahead approach, however, can result in trajectories that, while game-theoretically optimal, may not necessarily be trajectories that a human would pick in the first place. This shortsightedness is mitigated by search-based approaches which, unfortunately, are susceptible to risk-aware planning. A hybrid combination of both approaches is an interesting direction.
\end{enumerate}






{\footnotesize
\bibliographystyle{IEEEtran}
\bibliography{refs}
}

\end{document}